\documentclass[12pt,journal,compsoc]{IEEEtran}
\ifCLASSOPTIONcompsoc
\else
\fi
\ifCLASSINFOpdf
\else
\fi            
\usepackage[cmex10]{amsmath}
\usepackage{amssymb,bm,bbm,amsfonts,nicefrac}
\usepackage{algorithmic}
\usepackage{array}
\usepackage{dblfloatfix}
\usepackage{url}

\usepackage[dvips]{epsfig}
\usepackage{multirow}

\def\cqfd{\hfill\hbox{$\hbox{\vrule width 0.8pt
\vbox to6pt{\hrule depth 0.8pt width 5.2pt
\vfill\hrule depth 0.8pt}\vrule width 0.8pt}$}}

\newcommand{\defeq}{\stackrel{\mathrm{.}}{=}}

\newcommand{\matrice}[1]{\mathrm{#1}}

\newcommand{\proj}{{\mathbb{P}}}
\newcommand{\boundary}{{\mathbb{B}}}

\newcommand{\parent}[1]{\left({#1}\right)}

\newtheorem{theorem}{Theorem}
\newtheorem{lemma}{Lemma}
\newtheorem{corollary}{Corollary}

\newtheorem{definition}{Definition}

\newenvironment{proof}[1][Proof]{\begin{trivlist}
\item[\hskip \labelsep {\bfseries #1}]}{\end{trivlist}}

\newcommand\blfootnote[1]{%
  \begingroup
  \renewcommand\thefootnote{}\footnote{#1}%
  \addtocounter{footnote}{-1}%
  \endgroup
}

\def\dtheta{\mathrm{d}\theta}
\def\deta{\mathrm{d}\eta}

\newcommand{\gbot}{g_\bot}

\newcommand{\struct}[3]{(#1,#2)_{#3}}

\begin{document}

\title{On Conformal Divergences and their Population
  Minimizers$^\ddagger$}

\author{Richard~Nock,
Frank Nielsen,~
\and Shun-ichi Amari~
\IEEEcompsocitemizethanks{\IEEEcompsocthanksitem R. Nock (\textbf{contact author}) is
  with NICTA, Australia.
E-mail:richard.nock@nicta.com.au.\protect\\
F. Nielsen is with Ecole Polytechnique, France.
E-mail:nielsen@lix.polytechnique.fr.\protect\\
S.-i. Amari is with Riken Brain Science Institute, Japan.
E-mail:amari@brain.riken.jp.}
}

\IEEEtitleabstractindextext{%
\begin{abstract}
Total Bregman divergences are a recent tweak of ordinary Bregman
divergences originally motivated by applications that
required invariance by rotations.
They have displayed superior results compared to ordinary Bregman
divergences on several
clustering, computer vision, medical imaging and machine learning tasks.
These preliminary results raise two important problems :
First, report a complete characterization of the left and right
population minimizers for this class of total Bregman divergences.
Second, characterize a principled superset of total and ordinary
Bregman divergences with good
clustering properties, from which one could tailor the choice of a
divergence to a particular application.
In this paper, we provide and study one such superset with interesting
geometric features, that we
call conformal divergences, and focus on their left and right population
minimizers. Our results are
obtained in a recently coined $(u, v)$-geometric structure that is a
generalization of the dually flat affine
connections in information geometry. We characterize both analytically and
geometrically the population minimizers.
We prove that conformal divergences (resp. total Bregman divergences) are
essentially exhaustive for their left
(resp. right) population minimizers. We further report new results and
extend previous results on the robustness
to outliers of the left and right population minimizers, and discuss the
role of the $(u, v)$-geometric structure
in clustering. Additional results are also given.
\end{abstract}

}

\maketitle              

\IEEEdisplaynontitleabstractindextext
\IEEEpeerreviewmaketitle


\section{Introduction}

Loosely defined in
its most general form, the clustering problem is related to the grouping of a data sample according to
\textit{unknown} classes or clusters
\cite{vwgCS}\blfootnote{\hspace{-0.58cm}$^\ddagger$ Accepted for publication in IEEE Transactions on
  Information Theory, 2015.}. One of the most popular and well-posed approaches to
clustering is centroid-based: it seeks to summarize data into a fixed set of cluster
centers --- or \textit{population minimizers} ---
that best describe the sample, where "best" is understood with
respect to an expected measure of distortion to the whole sample
\cite{dlrML,mqSM}. Because of their convexity properties and links to
likelihoods in exponential families, ordinary Bregman divergences 
are often used to compute these distortions
in clustering algorithms, such as in $k$-means and EM
\cite{bmdgCW,dlrML,mqSM}. Their
left and right population minimizers are respectively
a cluster's $f$-mean \cite{kSL,nUE} and the cluster's average (ordinary Bregman divergences are in
general not symmetric). It has been shown that modulo technical
assumptions, ordinary Bregman divergences are exhaustive for their right
population minimizer: any divergence whose population minimizer is the sample
average is a Bregman divergence
\cite{afAC,bgwOT}. This shows that the scope of $k$-means and EM is
wide and 
encompasses all domains whose ``natural'' distortion measures rely on
Bregman divergences, such as signal processing, Euclidean geometry,
information theory, statistics, etc. .

There has been a recent burst of interest in a new class of
divergences, built from
Bregman divergences, known as total Bregman divergences
\cite{lTB,ehllIT,elhTB,llyyhCT,lvRA,lvanTB,lvdAR,vlanTB,lvanSR,rglTB}. These divergences
are invariant to particular transformations of the natural space. Experimentally
speaking, clustering with their left population minimizers yields significantly improved
results compared to ordinary Bregman divergences in domains like DTI
interpolation and segmentation \cite{lTB}. 
These
results exploit the fact that the left population minimizers of total Bregman divergences are weighted
generalized $f$-means \cite{lTB}. In the general context of clustering, and
also for reasons related to statistics and maximum likelihood
estimation \cite{svOB}, it is important to
characterize further the population minimizers of total Bregman
divergences: important questions include the characterization of
their right population minimizers and the exhaustivity of these
divergences for their population minimizers.
In the context of clustering, the
results of \cite{lTB} also contribute to the advocacy that clustering is in
fact a domain
dependent method \cite{vwgCS}, thereby raising the question of how we
may generalize further the set of candidate (total or ordinary)
Bregman divergences, while keeping good properties,
from which one may select the best candidates to solve a particular problem.

In this paper, we address these questions
in a setting which generalizes
in two ways total and ordinary Bregman divergences. First, we consider a superset
of total Bregman divergences and ordinary Bregman divergences that we define as conformal
divergences. Second, we consider a coordinate system which is not the
usual dually flat affine coordinate system of (total, ordinary) Bregman
divergences, but a generalization in information geometry studied by
Zhang and Amari, defined as the
$(u,v)$-geometric structure \cite{aTQ,aIG,zDF}, in which two coordinate
mappings $u$ and $v$ define the gradient (and its reciprocal inverse) of the generator of
the Bregman
divergence.

In this generalized setting, our main contribution includes:
\begin{itemize}
\item the characterization of the right population minimizers for
  total Bregman divergences;
\item the characterization of the right population minimizer for an interesting
  $L_p$ generalization of total Bregman divergences;
\item a proof that conformal divergences are exhaustive for their
  left population minimizers;
\item a proof that total Bregman divergences are exhaustive for
  their right population minimizers;
\item the robustness analysis of the left and right population
  minimizers for conformal divergences, which generalizes results
  known for total Bregman divergences \cite{lTB}.
\end{itemize}
Our contribution also includes results pertinent for clustering, such as
(i) a proof that the $(u,v)$-geometric structure sometimes describe an equivalence
relation which might be useful in the context of clustering; (ii) a
proof that the square loss in $v$-coordinates is the only 1D symmetric conformal
divergence in the $(u,v)$-geometric structure; (iii) a discussion on population minimizers for a further
extension involving the recently coined scaled Bregman divergences
(that generalize Csisz\'ar's $f$-divergences) \cite{svOB}.

The paper is structured as follows. The following Section gives 
definitions. Section \ref{sod} compares the various notions of
divergences we consider.
Section \ref{slp} is devoted to left population
minimizers of conformal divergences in the $(u,v)$-geometric
structure.  Section \ref{srpm} does the same for right population
minimizers. Section \ref{sro} studies the robustness of the population
minimizers and Section \ref{sdis} discusses our results. A last
section concludes. In order not to laden the paper's body, some proofs
are given in an Appendix in Section \ref{appen}. The rest of the
proofs, not in the published version, appear from page
\pageref{subright} in this extended version.

\section{Definitions}

Throughout this paper, bold faces denote column vectors, such as $\bm{0}$ for
the null vector, while capitals, like $J$ or $\mathrm{H}$ (respectively Jacobian and Hessian) denote
matrices. Coordinates are noted in exponent, such as $x^1, x^2,
..., x^d$ for vector $\bm{x} \in {\mathbb{R}}^d$, where $d\geq 1$. 

A (right-sided) \textit{conformal divergence}, $D_{\varphi, g}$, is parameterized by
two real-valued functions $\varphi$ and $g$ with $\mathrm{im} g
\subseteq (0, +\infty)$, whose domains are a compact
convex of ${\mathbb{R}}^d$. The expression of $D_{\varphi, g}$ is:
\begin{eqnarray}
D_{\varphi,g} (\bm{x} : \bm{y}) & \defeq & g(\bm{y})D_\varphi(\bm{x}:\bm{y})\:\:.\label{defcd}
\end{eqnarray}
$\varphi$ is real-valued strictly convex twice differentiable, and
$D_\varphi(\bm{x}:\bm{y})$ is the ordinary Bregman divergence with
generator $\varphi$:
\begin{eqnarray}
D_\varphi(\bm{x}:\bm{y}) & = & \varphi(\bm{x}) - \varphi(\bm{y}) -
(\bm{x}-\bm{y})^\top \nabla\varphi(\bm{y})\:\:.\label{defbreg}
\end{eqnarray}
$\nabla\varphi$ denotes the gradient of $\varphi$. 
$g$ admits
continuous directional derivatives: function $\mathrm{D}_{{\bm{z}}}
g(\bm{x}) \defeq \lim_{t\rightarrow 0} \mathrm{D}_{t,{\bm{z}}}
g(\bm{x})$, defining directional derivatives, is continuous and exist
for any \textit{valid} direction $\bm{z}$ such that $\mathrm{D}_{t, \bm{z}} g(\bm{x})$ is defined in a
neighborhood of 0 (with respect to $t$). We give:
\begin{eqnarray}
\mathrm{D}_{t, \bm{z}} g(\bm{x}) & \defeq & \frac{g(\bm{x} +
  t\bm{z}) - g(\bm{x})}{t}\:\:.\nonumber
\end{eqnarray}
Ordinary Bregman divergences match the subset of conformal divergences for
which $g(.) = K$, a constant. The most popular recent example of conformal divergences is
obtained for $g = K \gbot$ for some constant $K>0$ and :
\begin{eqnarray}
\gbot(\bm{y}) & \defeq & \frac{1}{\sqrt{1+\|\nabla \varphi(\bm{y})\|_2^2}}\:\:,\label{refg2}
\end{eqnarray}
which defines total Bregman divergences, that are invariant to rotations of the coordinate
axes \cite{ehllIT,llyyhCT,lvRA,lvanSR,lvanTB,lvdAR,vlanTB} (among
others). Table \ref{overphit} presents some examples of total Bregman
divergences (with $K=1$). Remark
that $\gbot(\bm{y})$ is of the form $f_\bot(\nabla\varphi(\bm{y}))$,
with 
\begin{eqnarray}
f_\bot(\bm{x}) & \defeq &
\frac{1}{\sqrt{1+\|\bm{x}\|_2^2}}\:\:.\label{deffbot}
\end{eqnarray} 
Figure \ref{depic} depicts
$D_\varphi(x:y)$ and $D_{\varphi,\gbot}(x:y)$ on a simple example.
\begin{figure}[t]
\centering
\begin{tabular}{c} 
\epsfig{file=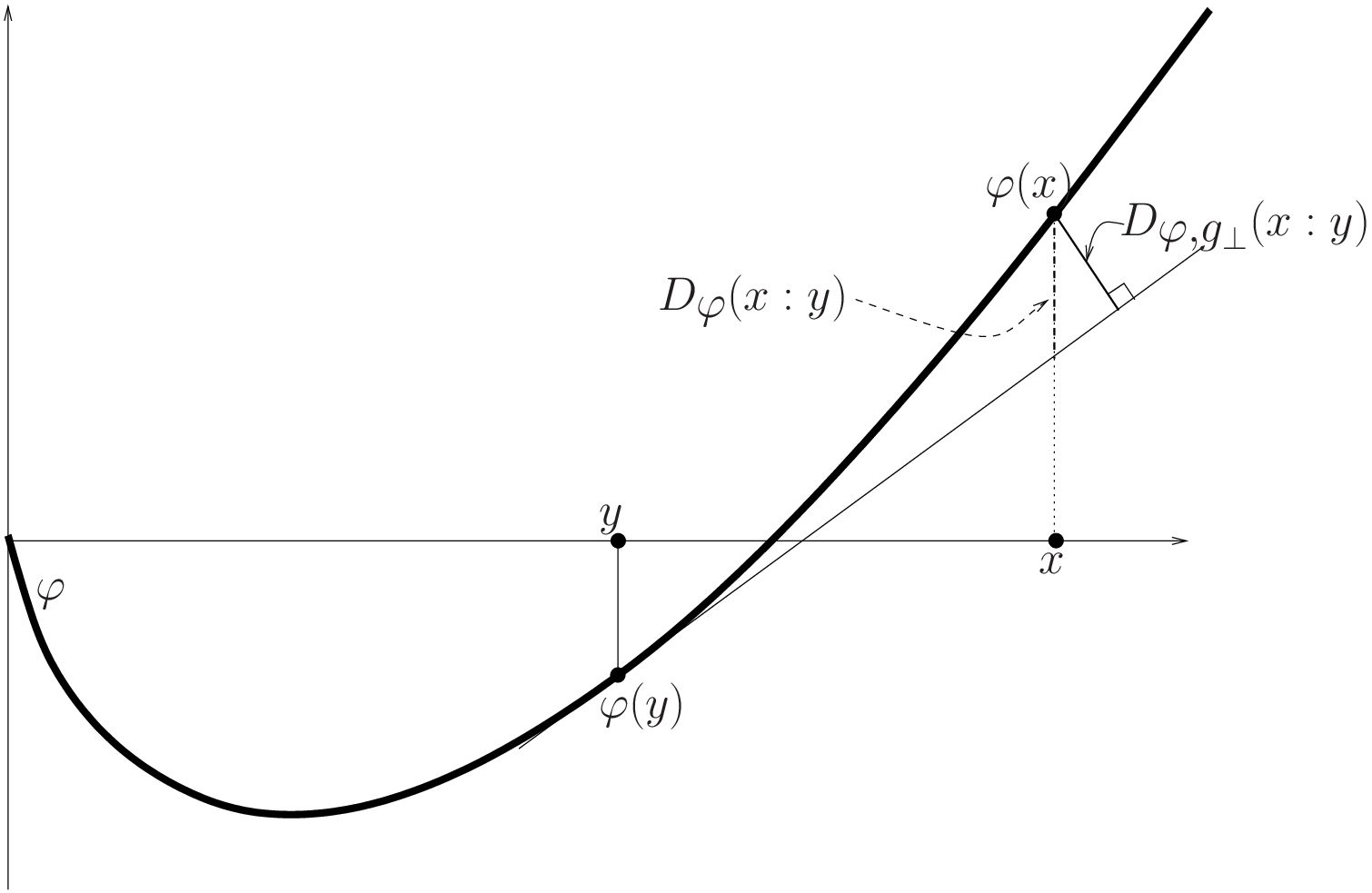, width=0.9\columnwidth}
\end{tabular}
\caption{Depiction of $D_\varphi(x:y)$ and $D_{\varphi,\gbot}(x:y)$
  when $\varphi(x) = x\ln(x)-x$.}\label{depic}
\end{figure}
We also
investigate the generalization of (\ref{refg2}) to $p$-norms, and
define, $\forall p\geq 1$:
\begin{eqnarray}
g_p(\bm{y}) & \defeq & f_p(\nabla\varphi(\bm{y})) \:\:;\label{defgp}\\
f_p(\bm{x}) & \defeq & \frac{1}{\left(1+\|\bm{x}\|_p^p\right)^{\frac{1}{p}}} \:\:.\label{deffp}
\end{eqnarray}
The $p$-norm of $\bm{x}$ is $\|\bm{x}\|_p \defeq \left(\sum_i
  |x^i|^p\right)^\frac{1}{p}$.

A coordinate mapping $v$ is a $C^1$, bijective function $v : {\mathbb{R}}^d
\rightarrow {\mathbb{R}}^d$. For any coordinate mapping $v$, we define the $v$-conformal divergence $D^v_{\varphi, g}$ as:
\begin{eqnarray}
D^v_{\varphi,g} (\bm{x} : \bm{y}) & \defeq & g(\bm{y})D_\varphi(v(\bm{x}):v(\bm{y}))\:\:.\label{defucd}
\end{eqnarray}
$v$-conformal divergences are inspired by divergences in the $(u,v)$
geometric structure \cite{aTQ,aIG} (see also Section \ref{sdis}). They generalize conformal divergences
for which $v = \mathrm{Id}$. We shall investigate several interesting
cases of $v$-conformal divergences, including those where $g$ is a
function of $\nabla\varphi$, and those where $g$ is a function of 
coordinate mapping $u$ in the $(u,v)$-geometric structure.

Let us now motivate the $(u,v)$-geometric structure in the context of the dual
coordinate systems of ordinary Bregman divergences \cite{anMO}.
Function $g$ in $v$-conformal divergences depends on the right
parameter of the divergence. We shall see (Lemma \ref{lsd}) that when
$g$ is not constant, the $v$-conformal divergence cannot be symmetric: $g(\bm{y})D_\varphi(v(\bm{x}):v(\bm{y})) \neq
g(\bm{y})D_\varphi(v(\bm{y}):v(\bm{x}))$. However, our results extend
at little cost to left-sided conformal divergences, \textit{i.e.}
whose regularization factor $g$ depends
on the left parameter of the divergence. Indeed, calling to convex
conjugates, we obtain:
\begin{eqnarray*}
\lefteqn{D_{\varphi, g}^v (\bm{x} : \bm{y})}\nonumber\\
 & = & g(\bm{y}) D_{\varphi}
(v(\bm{x}) : v(\bm{y}))\\
 & = & g(\bm{y}) D_{\varphi^\star}
((\nabla\varphi \circ v)(\bm{y}) : (\nabla\varphi \circ v) (\bm{x}))\\
 & = & g(\bm{y}) D_{\varphi^\star}
(u (\bm{y}) : u (\bm{x}))\:\:,
\end{eqnarray*}
where $u \defeq \nabla\varphi \circ v$ also defines a coordinate mapping and
$\varphi^\star$ is the convex conjugate of $\varphi$. Any
such coordinate mappings $u$ and $v$ such that $u \circ v^{-1}$
defines the gradient of a strictly convex differentiable function $\varphi$ is
called an $(u, v)$-geometric structure \cite{aTQ}, that we 
write $\struct{u}{v}{\varphi}$ from now on, to make explicit the reference to
$\varphi$. 

We now define population minimizers for conformal divergences.

\begin{definition}(left- and right-population minimizers)
Let ${\mathcal{S}} \defeq \{\bm{x}_1, \bm{x}_2, ...,
\bm{x}_n\}$, with $\bm{x}_i \in {\mathbb{R}}^d, \forall i = 1, 2, ...,
n$. Let $D : {\mathbb{R}}^d\times  {\mathbb{R}}^d \rightarrow
{\mathbb{R}}$ be the shorthand for some ordinary Bregman (resp. conformal,
resp. $v$-conformal) divergence $D_\varphi$ (resp. $D_{\varphi, g}$,
resp. $D_{\varphi, g}^v$). A \textit{left population minimizer} for
$D$ on ${\mathcal{S}}$ is any $\bm{\mu}$ such that
$\sum_i{D(\bm{\mu} : \bm{x}_i)} = \min_{\bm{x}} \sum_i D(\bm{x} :
\bm{x}_i)$. A \textit{right population minimizer} for
$D$ on ${\mathcal{S}}$ is any $\bm{\mu}$ such that
$\sum_i{D(\bm{x}_i : \bm{\mu})} = \min_{\bm{x}} \sum_i D(
\bm{x}_i : \bm{x})$.
\end{definition}
This definition, as well as the results in this paper,
can be extended to non uniform distributions over ${\mathcal{S}}$, and
to population minimizers in the continuous case. 

\section{On divergences: ordinary, total and conformal}\label{sod}

There is a need to generalize the source of divergences from which
efficient centroid-based clustering algorithms may be derived. The
comparison between ordinary Bregman and total Bregman divergences is enlightening from that standpoint:
ordinary Bregman divergences $D_\varphi$ are \textit{axiomatically}
characterized as the unique family of divergences (under mild
conditions) that yield their right population minimizers matching the
\textit{sample average} \cite{bgwOT}. Hence, the right population
minimizer is the simplest to compute, \textit{but} having fixed the data
sample, 
regardless of the
generator of the divergence $\varphi$, it is always the 
\textit{same}. From a clustering
standpoint, it may be more intuitive that since changing the
generator changes the geometry of the problem, it should possibly
change this population minimizer as well. Also, this invariance
is not convenient to further optimize the population minimizer by
tuning the divergence at hand.

This problem does not appear anymore with total Bregman
divergences. Initially, total Bregman divergences \cite{lTB} $D_{\varphi,\gbot}$
have been geometrically designed to enforce invariance by rotations in
the parameter space \cite{lTB}, thus mimicking the ordinary/total
least squares relationships.
Rotation invariance is a very desirable property in medical imaging \cite{vlanTB} and
computer vision \cite{rglTB}. Thus, total Bregman divergences have
been specifically engineered to solve a particular geometric problem,
which has led to improved results on several key applications related to clustering.
Besides, total Bregman divergences have also proven
\textit{experimentally}  superior in boosting \cite{lvRA} and
tensor-based graph matching \cite{elhTB}, etc., just to name a few.
One theoretical argument that explains the superiority of total
Bregman divergences was detailed in \cite{lvanSR}, where it was proved
that total Bregman divergences are robust compared to ordinary Bregman
divergences, by studying the impact of outliers via the influence function.

The difference between total and ordinary Bregman divergences can also
be
captured from a statistical standpoint.
It is well-known that regular exponential families
$p(\bm{x};\bm{\theta})=h(\bm{x})\exp(\bm{\theta}^\top
\bm{t}(\bm{x})-\varphi(\bm{\theta}))$ are in bijection with (regular)
Bregman divergences \cite{bmdgCW},
$p(\bm{x};\bm{\theta})=h(\bm{x})\exp(-D_{\varphi^\star}(\bm{t}(\bm{x}):\bm{\eta}(\bm{\theta})))$
where $\bm{\eta}(\bm{\theta})=\nabla\varphi(\bm{\theta})$ is the dual moment parameter, and that
the Maximum Likelihood Estimator (MLE) for $n$ identically and
independently distributed observations of an exponential family
coincides with the so-called \textit{observed point} in information
geometry \cite{anMO}:  $\overline{\bm{t}}=(1/n)\sum_i
\bm{t}(\bm{x}_i)$, where $\bm{t}(.)$ denotes the vector of sufficient
statistics of the exponential families under consideration.  
That is, the MLE $\hat{\bm{\eta}}$ expressed in the $\bm{\eta}$-parameter matches the centroid of sufficient statistics: $\hat{\bm{\eta}}=\overline{\bm{t}}$.
Since there is also a bijection between ordinary and total Bregman
divergences, we deduce by transitivity with the ordinary Bregman-exponential family bijection that we can associate an exponential family 
$p(\bm{x};\bm{\theta})=h(\bm{x})\exp(-D_{\varphi^\star,\gbot}(\bm{t}(\bm{x}):\bm{\eta})\sqrt{1+\|\nabla
  \varphi^\star(\bm{\eta})\|^2})$ to any total Bregman divergence $D_{\varphi^\star,\gbot}$.
This statistical distribution $p(\bm{x};\bm{\theta})$ corresponds also to a
\textit{lifted exponential family}
$\tilde{p}(\tilde{\bm{x}};\tilde{\bm{\theta}})$ in disguise, as
exemplified in \cite{lvanSR} with
$\tilde{\bm{\theta}}=(\bm{\theta},\varphi(\bm{\theta}))$ and
$\tilde{\bm{x}}=(1/\sqrt{1+\|\nabla \varphi(\bm{x})\|^2}) \cdot (\bm{x},1)$.
In other words, the ordinary exponential family is lifted to the space having one extra dimension and embedded  as a hypersurface.
Now, it can be proved that the \textit{total (left) observed point} is the Bayesian MAP estimator in the lifted exponential family  with prior distribution $\pi(\bm{\theta})=\exp(-n \tilde{\varphi}(\bm{\theta}))$,
where $\tilde{\varphi}(\bm{\theta})$ is the normalization factor.
Table~\ref{tab:compare} compares properties of ordinary and total Bregman divergences. "Information" relates to
the sample divergence to the right population minimizer \cite{bmdgCW}.

\begin{table*}
\begin{center}
\begin{tabular}{|c||c|c|}\hline
& ordinary Bregman divergence & total Bregman divergence \\ \hline\hline
Right population minimizer & centroid &  this paper (Corollary \ref{cororth2}) \\ \hline
Left population minimizer & $\nabla \varphi$-mean, weights $w_i= 1/n$
&  $\nabla \varphi$-mean, weights $w_i= 1 / \sqrt{1+\|\nabla \varphi(\bm{x}_i)\|^2}$ \\ \hline
Robustness of pop. minimizers & no & yes \\ \hline
Information & Bregman information & this paper (Lemma \ref{ld}) \\ \hline
Exhaustiveness & right pop. minimizer = sample average & this paper (Theorem \ref{thex}) \\ \hline
Bijection & exponential families & lifted exponential families\\ \hline
Inference & MLE = observed point  & Bayesian MAP = total observed point \\ 
& ($\hat{\bm{\eta}}$=right population minimizer) & (left population
minimizer) \\ \hline\hline
\end{tabular}
\end{center}
\caption{Properties of ordinary Bregman vs total Bregman
  divergences. 
\label{tab:compare}}
\end{table*}

Our definition of conformal divergences is inspired by information
geometry. In information geometry \cite{anMO}, a
\textit{divergence}~\cite{zDF,bDM} (also called a contrast function or
yoke) is a measure of dissimilarity $D(\bm{p}:\bm{q})$ (with
$D(\bm{p}:\bm{q})\geq 0$ with equality iff $p=q$, but not necessarily
symmetric nor satisfying the triangle inequality) that further needs
to satisfy some smoothness conditions \cite{zDF} to induce properly a
metric tensor and a cubic form (for the coefficients of the
connections).  
In Riemannian geometry, a \textit{conformal metric} $\textsl{g}'$ of a
metric $\textsl{g}$ is expressed by $\textsl{g}'=\varrho \textsl{g}$,
where $\varrho>0$ (hence also written as $\textsl{g}'=e^{\varrho}\textsl{g}$). 
The uniformization theorem states that Riemannian surfaces are
conformally equivalent to either the spherical, planar or hyperbolic
manifolds --- all of constant curvatures. In our definition of
conformal divergences in eq. (\ref{defucd}), factor $g(.)$ plays the role of the conformal
factor; we shall see in Lemma \ref{ld} below that it indeed defines a
conformal factor. In the particular case of total Bregman divergences, $\varrho \propto 1/\sqrt{1+\|\nabla \varphi(.)\|^2}$
plays the role of the conformal factor. 
In information geometry induced by a generalized logarithm function, a
conformal flattening \cite{omaAD} allows to obtain a dually flat structure.
Conformal mappings also explain the role of \textit{escort
  distributions}, and yield efficient algorithms for Voronoi
diagrams induced by conformal divergences, a geometric structure
particularly relevant to clustering \cite{omaAD,bnnBV}.

\section{Left population minimizers of $v$-conformal divergences}\label{slp}

We are interested in this Section in characterizing the left population
minimizers of general $v$-conformal divergences. We build on
results known from \cite{lvanSR,lvanTB} for the elicitation of the
left population minimizer when $v = \mathrm{Id}$, and the well known
results from \cite{bgwOT} for the elicitation of the divergences
having the arithmetic average as right population minimizer. Technicalities are simpler than for the right population minimizers
because function $g$ does not depend on the left parameter of
$D_\varphi$. We first show that the left population minimizer of some
$v$-conformal divergence $D_{\varphi, g}^v$ is a weighted $u$-mean,
where $\struct{u}{v}{\varphi}$ is a geometric structure.
\begin{lemma}\label{lleft1}
The left population minimizer $\bm{\mu}$ of any $v$-conformal divergence
$D_{\varphi, g}^v$ on ${\mathcal{S}}$ is unique and
equals:
\begin{eqnarray}
\bm{\mu} & = &  u^{-1}\left( \frac{1}{\sum_i
    g(\bm{x}_i)} \sum_i g(\bm{x}_i) u(\bm{x}_i)\right)\:\:,\label{propuv}
\end{eqnarray}
where $\struct{u}{v}{\varphi}$ is a geometric structure.
\end{lemma}
(Proof in Appendix, Subsection \ref{proof_lleft1})
We now show that the characterization of left population minimizers
for $v$-conformal divergences is exhaustive, as any distortion
function admitting a weighted $u$-mean as left population minimizer
equals a $v$-conformal divergence $D_{\varphi, g}^v$, for some
$\struct{u}{v}{\varphi}$-geometric structure.
\begin{lemma}\label{lleft2}
Let $\mu \defeq u^{-1}(\sum_i {w_i u(\bm{x}_i)})$ be the unique
solution to $\min_{\bm{x}} \sum_i {D(\bm{x} : \bm{x}_i)}$, where:
\begin{enumerate}
\item $D : {\mathbb{R}}^d\times {\mathbb{R}}^d
\rightarrow {\mathbb{R}}$ is non-negative, twice continuously differentiable and such that
$D(\bm{x} : \bm{x}) = 0,\forall \bm{x}$;
\item $u : {\mathbb{R}}^d
\rightarrow {\mathbb{R}}^d$ is a coordinate mapping;
\item $\sum_i w_i =
1$ and $w_i > 0, \forall i$. 
\end{enumerate}
Then there exist a function $g: {\mathbb{R}}^d
\rightarrow {\mathbb{R}}$ admitting continuous directional derivatives
and a geometric structure $\struct{u}{v}{\varphi}$ such that 
\begin{eqnarray}
D(\bm{x} :\bm{y}) & = & D^v_{\varphi, g}(\bm{x} :\bm{y})\:\:. \label{dd4}
\end{eqnarray}
\end{lemma}
(Proof in Appendix, Subsection \ref{proof_lleft2})


\section{Right population minimizers of $v$-conformal divergences}\label{srpm}

\subsection{Case $v = \mathrm{Id}$}

We now derive the right population minimizers for a conformal
divergence $D_{\varphi, g}$, thus considering $v$-conformal
divergences with $v = \mathrm{Id}$. 
Because $g$ admits continuous directional derivatives, so does $D_{\varphi,g} (\bm{x} :
\bm{y})$ for both its arguments. Let us
define:
\begin{eqnarray} 
\lefteqn{\mathrm{D}_{t,\bm{z}}
  D_{\varphi,g}(\bm{x}:\bm{y})}\nonumber\\
 & \defeq & \frac{D_{\varphi,g}(\bm{x}:\bm{y} + t \bm{z}) -
D_{\varphi,g}(\bm{x}:\bm{y})}{t}\:\:,\label{defzz1}
\end{eqnarray}
so that the directional derivative in the right parameter $\mathrm{D}_{\bm{z}} D_{\varphi,g}(\bm{x}:\bm{y}) \defeq
\lim_{t\rightarrow 0} \mathrm{D}_{t,\bm{z}}
D_{\varphi,g}(\bm{x}:\bm{y})$ exists, for any valid direction $\bm{z}$. Define from any ${\mathcal{S}}$ the following averages:
\begin{eqnarray}
\overline{\varphi} & \defeq & \frac{1}{n} \sum_i
\varphi(\bm{x}_i)\:\:,\label{defovarphi}\\
\overline{\bm{x}} & \defeq & \frac{1}{n} \sum_i
\bm{x}_i\:\:.\label{defox}
\end{eqnarray}
Let us define the following vectors $\overline{\bm{x}}^+, \bm{\mu}^+, \bm{\delta}^+, \bm{z}^+ \in {\mathbb{R}}^{d+1}$:
\begin{eqnarray}
\overline{\bm{x}}^+ \hspace{-0.3cm} & \defeq & \hspace{-0.3cm} \left[
\begin{array}{c}
\overline{\bm{x}} \\
\overline{\varphi} 
\end{array}
\right] \label{defxp}\:\:,\\
\bm{\mu}^+ \hspace{-0.3cm} & \defeq & \hspace{-0.3cm} \left[
\begin{array}{c}
\bm{\mu}\\
\varphi(\bm{\mu})
\end{array}
\right] \label{defmup}\:\:,\\
\bm{\delta}^+ \hspace{-0.3cm} & \defeq & \overline{\bm{x}}^+ -\bm{\mu}^+ \label{defdelta}\:\:,\\
\bm{z}^+ \hspace{-0.3cm} & \defeq & \hspace{-0.3cm} \left[
\begin{array}{c}
\mathrm{D}_{\bm{z}} (g(\bm{\mu}) \nabla \varphi(\bm{\mu}))\\
-\mathrm{D}_{\bm{z}} g(\bm{\mu})
\end{array}
\right] \label{defzz2}\:\:,
\end{eqnarray}
from which we define the following sets:
\begin{eqnarray}
\proj_{{\mathcal{S}},\varphi,g} & \defeq & \left\{\bm{\mu} \in
  \mathrm{dom}(D_{\varphi,g}) : \bm{\delta}^+ \bot \bm{z}^+, \forall
  \bm{z} 
\right\}\:\:,\label{defcorr1}
\end{eqnarray}
where the directions $\bm{z}$ have to be valid, \textit{i.e.} such
that the directional derivative of $g$ is defined in a
neighborhood of 0 (with respect to $t$). We also define
$\boundary_{\varphi,g}$, the eventually empty set of non-differentiable boundary
points of the intersection of the domains of $\varphi$ and $g$.
In the following, we let
$\mathcal{P}(D_{\varphi,g};{\mathcal{S}})$ denote the set of right population minimizers
for conformal divergence $D_{\varphi,
  g}$ on set ${\mathcal{S}}$.

\begin{lemma}\label{ld}
$\mathcal{P}(D_{\varphi,g};{\mathcal{S}}) \subseteq
\proj_{{\mathcal{S}},\varphi,g} \cup \boundary_{\varphi,g}$. Furthermore,
$\forall \bm{\mu} \in \proj_{{\mathcal{S}},\varphi,g}
\backslash \{\overline{\bm{x}}\}$, the
average distortion is a weighted square Mahalanobis distance to the
population average:
\begin{eqnarray}
\lefteqn{\frac{1}{n} \sum_i
D_{\varphi,g}(\bm{x}_i:\bm{\mu})}\nonumber\\
 & = & \varrho_g \times (\overline{\bm{x}} -
\bm{\mu})^\top \mathrm{H}\varphi(\bm{\mu}) (\overline{\bm{x}} -
\bm{\mu})\:\:, \label{valdivmaha}
\end{eqnarray}
with $\varrho_g \defeq g^2(\bm{\mu}) / \mathrm{D}_{(\overline{\bm{x}} -
\bm{\mu})}
g(\bm{\mu}) > 0$.
\end{lemma}
(Proof in Appendix, Subsection \ref{proof_ld}) With respect to the
discussion on conformal divergences in Section \ref{sod}, we see that
$\varrho_g$ defines a conformal factor, and so all points in $\proj_{{\mathcal{S}},\varphi,g}
\backslash \{\overline{\bm{x}}\}$ are points for which the conformal
divergence reduces to a (square) distance on metric
$\mathrm{H}\varphi$ conformally
transformed. We shall see that set $\proj_{{\mathcal{S}},\varphi,g}
\backslash \{\overline{\bm{x}}\}$ also contains population
minimizers. Finally, Lemma \ref{ld} shows that the conformal Bregman
information generalizes the Bregman information \cite{bmdgCW} to
\textit{weighted} square Mahalanobis distance --- because Bregman divergences can be
formulated as square Mahalanobis distance over particular metrics, Bregman
information can also be expressed using square Mahalanobis distance.

We are now ready to
state a first Theorem that provides the right population minimizers
for a subset of conformal divergences which encompasses total
Bregman divergences. 
\begin{theorem}\label{thh0}
Pick $g = K g_p$ as in (\ref{defgp}) with $K>0$ a constant, $p = 2k/(2k-1)$ and $k \in
{\mathbb{N}}_*$. Then, assuming $\boundary_{\varphi,g} = \emptyset$, the right population
minimizer(s) for $D_{\varphi,g_p}$ on ${\mathcal{S}}$ match the set:
\begin{eqnarray}
\mathcal{P}(D_{\varphi,K g_p};{\mathcal{S}}) & = & \arg\min_{\bm{\mu}} \|\overline{\bm{x}}^+ - \bm{\mu}^+\|_q\:\:,
\end{eqnarray}
with $\overline{\bm{x}}^+$ and $\bm{\mu}^+$ defined in (\ref{defxp})
and (\ref{defmup}), and $q  = 2k \in {\mathbb{N}}$ is the H\"older conjugate of
$p$\footnote{Two reals $p, q\geq 1$ are H\"older conjugates when $(1/p) + (1/q) = 1$.}.
\end{theorem}
(Proof in Appendix, Subsection \ref{proof_thh0}) 

Fixing $k=1$ allows
to retrieve the right population minimizers for total Bregman
divergences. Because of their importance, we state their
characterization as a separate corollary. 
\begin{corollary}\label{cororth2}
Consider $g=K g_\bot$ for some constant $K>0$. The following holds true:
\begin{eqnarray}
\bm{\delta}^+ & \bot & \left[\begin{array}{c}
\nabla\varphi(\bm{\mu})\\
\|\nabla
\varphi(\bm{\mu})\|_2^2
\end{array}
\right]\:\:, \forall \bm{\mu} \in
\proj_{{\mathcal{S}},\varphi,g_\bot}\:\:;\label{orthmuu1}
\end{eqnarray}
\textit{i.e.}, the orthogonal projection of $(\overline{\bm{x}},
\overline{\varphi})$ on the tangent hyperplane $T_\varphi(\bm{\mu})$
to $\varphi$ at $\bm{\mu}$ is point $(\bm{\mu},
\varphi(\bm{\mu}))$. Furthermore, assuming $\boundary_{\varphi,g} = \emptyset$, we have:
\begin{eqnarray}
\mathcal{P}(D_{\varphi, K g_\bot};{\mathcal{S}}) & = & \arg\min_{\bm{\mu}} \|\overline{\bm{x}}^+ -
\bm{\mu}^+\|_2\:\:, \label{propp1}
\end{eqnarray}
with $\overline{\bm{x}}^+$ and $\bm{\mu}^+$ defined in (\ref{defxp})
and (\ref{defmup}).
\end{corollary}
Figure \ref{overphifig} (left) displays how to find $\bm{\mu}$
which meets condition (\ref{orthmuu1}). Notice that, by construction,
the right population
minimizer for $D_{\varphi,K \gbot}$ is invariant by rotation of
the axes. Figure \ref{overphifig} (right) depicts the construction of the
population minimizer in a simple 1D case.

\begin{figure}[t]
\centering
\begin{tabular}{cc} 
\epsfig{file=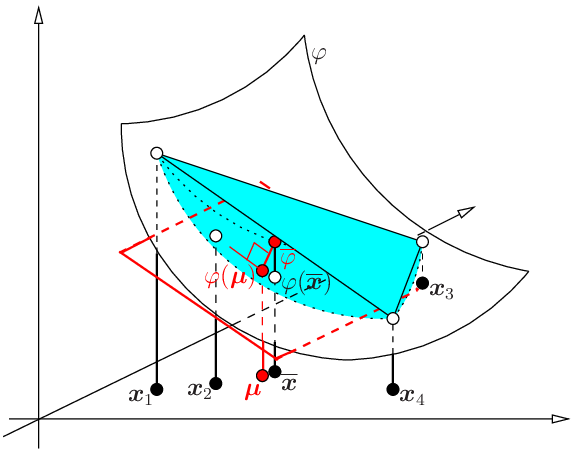, width=0.45\columnwidth} & \epsfig{file=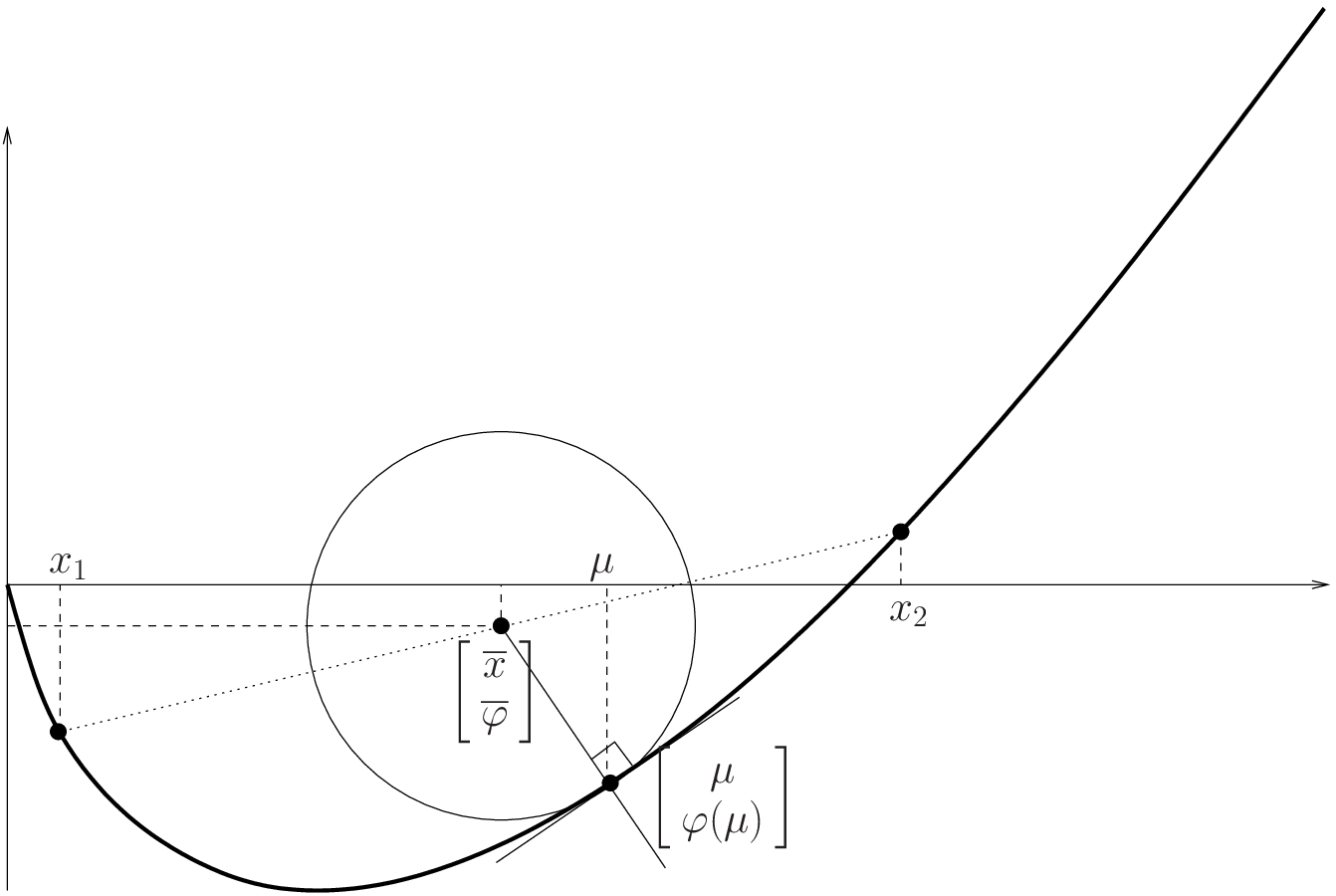, width=0.45\columnwidth}
\end{tabular}
\caption{Left: how to find $\bm{\mu}$ as stated in Corollary \ref{cororth2}:
  the orthogonal projection of point $(\overline{\bm{x}},
  \overline{\varphi})$ on the hyperplane $T_\varphi(\bm{\mu})$ tangent
  to $\varphi$ at $\bm{\mu}$ coincides
  with point $(\bm{\mu}, \varphi(\bm{\mu}))$. The blue region depicts
  the subset of the epigraph of $\varphi$ which is below the image
  by $\varphi$ of the convex
  envelope of ${\mathcal{S}}$. Right: computation of the (unique) population minimizer on a simple
  1D example with $\varphi(x) = x\ln x - x$, following Corollary
  \ref{cororth2}. Remark that $\mu > \overline{x}$ in this case.}\label{overphifig}
\end{figure}

We now study to what extent total Bregman divergences
are exhaustive for the construction of the right population
minimizer depicted in Corollary \ref{cororth2}. It has been shown that
ordinary Bregman divergences are
exhaustive for the expectation as right population minimizer,
\textit{i.e.} if the expectation is the right population minimum of a
loss $D(x:y)$, then under mild conditions this loss is an ordinary Bregman
divergence \cite{afAC,bgwOT}. It turns out that total Bregman
divergence are also exhaustive for their right population
minimizer. For the sake of simplicity, we are going to show the result
in the one-dimensional setting ($d=1$). For this objective, we let:
\begin{eqnarray}
\lefteqn{\proj_{{\mathcal{S}},\varphi}}\nonumber\\
 & \defeq & \left\{\mu \in
  {\mathbb{R}}: 
\left[
\begin{array}{c}
\overline{x} -
\mu\\
\overline{\varphi} -
\varphi(\mu)
\end{array}
\right] \bot 
\left[
\begin{array}{c}
1\\
\varphi'(\mu)
\end{array}
\right]
\right\}\:\:.\label{defcorr2}
\end{eqnarray}
When $\mu \neq \overline{x}$, the condition is equivalent to
\begin{eqnarray}
\tilde{\varphi}'_{\mathcal{S}}(\mu) \varphi'(\mu) & = & -1\:\:,\label{defcorr33}
\end{eqnarray}
with
\begin{eqnarray}
\tilde{\varphi}'_{\mathcal{S}}(z) & \defeq & \frac{\overline{\varphi} - \varphi(z)}{\overline{x} -
    z}\:\:, \forall z \in \mathrm{dom}(\varphi)\backslash
  \{\overline{x}\}\:\:,\label{defyz}
\end{eqnarray}

\begin{theorem}\label{thex}
Let $D :
{\mathbb{R}}\times {\mathbb{R}} \rightarrow {\mathbb{R}}$ be a
function differentiable such that $D(x:y)$ is twice continuously
differentiable in $x$, and satisfies the following assumptions: (i)
$D(x:x) = 0,\forall x$, (ii) $D(x:y) > 0, \forall y\neq x$, (iii) $D$
is invariant by rotation of the axes, (iv) the
right population minimizer of $D$ on ${\mathcal{S}} = \{x_1, x_2, ...,
x_n\}$ is unique and satisfies
$\{\mu\} \in \proj_{{\mathcal{S}},\varphi}$ for some strictly convex twice
differentiable $\varphi$. Then 
\begin{eqnarray}
D(x:y) & = & D_{\varphi, K \gbot}(x:y)\:\:,\label{pvphi}
\end{eqnarray}
where $K>0$ is a constant and
$\gbot$ is defined in eq. (\ref{refg2}).
\end{theorem}
(proof in Subsection \ref{subright})

\subsection{Case $v$ arbitrary}

We now focus on general $v$-conformal divergences with $\struct{u}{v}{\varphi}$ a
geometric structure. In order not to laden this Section and its notations, we make the simplifying assumption
that $\boundary_{\varphi,g} = \emptyset$. This is not restrictive: even in the multidimensional extension of the total Bregman
divergences of Table \ref{overphit}, the cardinal of
$\boundary_{\varphi,g} = \emptyset$ would be at most one, so the main
structural and algorithmic issues to characterize the right population
minimizers essentially lie in the characterization of $\proj_{{\mathcal{S}},\varphi}$.
Define from ${\mathcal{S}}$ the following averages:
\begin{eqnarray}
\overline{\varphi}_v & \defeq & \frac{1}{n}\sum_i
  \varphi(v(\bm{x}_i))\label{defovarphi2}\\
 & = & \frac{1}{n}\sum_i
  u(\bm{x}_i)^\top v(\bm{x}_i) - \overline{\varphi^\star}_u\:\:;\nonumber\\
\overline{\bm{x}}_v & \defeq & \frac{1}{n}\sum_i
  v(\bm{x}_i)\:\:. \label{defox2}
\end{eqnarray}
We first state a
generalization of Theorem \ref{thh0} to arbitrary
$\struct{u}{v}{\varphi}$-geometric structures.
\begin{theorem}\label{thh0uv}
Let $\struct{u}{v}{\varphi}$ be a geometric structure. Pick $g = K g_p^u
(\bm{\mu})$ with $g_p^u
(\bm{\mu}) \defeq f_p(u(\bm{\mu}))$, $f_p$ is defined in
(\ref{deffp}), $K>0$ is a constant, and $p = 2k/(2k-1)$ with $k \in
{\mathbb{N}}_*$. Then, the right population
minimizer(s) for the $v$-conformal divergence $D^v_{\varphi,K g^u_p}$ on ${\mathcal{S}}$ match the set:
\begin{eqnarray}
\mathcal{P}(D^v_{\varphi,K g^u_p};{\mathcal{S}}) & = & \arg\min_{\bm{\mu}} \|\overline{\bm{x}}^+_v - \bm{\mu}^+_v\|_q\:\:,
\end{eqnarray}
with $\overline{\bm{x}}^+_v$ and $\bm{\mu}^+_v$ defined in (\ref{defxpu})
and (\ref{defmupu}), and $q  = 2k \in {\mathbb{N}}$ is the H\"older conjugate of
$p$.
\end{theorem}
(Proof omitted) We also provide the following generalization of
Corollary \ref{cororth2}, which stands as a Corollary to Theorem \ref{thh0uv}.
\begin{corollary}\label{th2phi}
Let $\struct{u}{v}{\varphi}$ be a geometric structure. Pick
$g(\bm{\mu})=K g^u_\bot(\bm{\mu}) \defeq K f_\bot(u(\bm{\mu}))$ for
any constant $K>0$. Any right population minimizer $\bm{\mu}$ for the $v$-conformal
divergence $D^v_{\varphi,K g^u_\bot}$ satisfies:
\begin{eqnarray}
\left[
\begin{array}{c}
\overline{\bm{x}}_v-v(\bm{\mu})\\
\overline{\varphi}_v - \varphi(v(\bm{\mu}))
\end{array}
\right]   & \bot & \left[\begin{array}{c}
u(\bm{\mu})\\
\|u(\bm{\mu})\|_2^2
\end{array}
\right]\:\:. \label{orthmuu2}
\end{eqnarray}
Furthermore, we have:
\begin{eqnarray}
\mathcal{P}(D^v_{\varphi, K g^u_\bot};{\mathcal{S}}) & = & \arg\min_{\bm{\mu}} \|\overline{\bm{x}}_v^+ -
\bm{\mu}_v^+\|_2\:\:, \label{proppu1}
\end{eqnarray}
with $\overline{\bm{x}}^+_v$ and $\bm{\mu}^+_v$ defined as follows:
\begin{eqnarray}
\overline{\bm{x}}^+_v \hspace{-0.40cm} & \defeq & \hspace{-0.40cm} \left[
\begin{array}{c}
\overline{\bm{x}}_v \\
\overline{\varphi}_v 
\end{array}
\right] = \left[
\begin{array}{c}
\overline{\bm{x}}_v \\
\frac{1}{n}\sum_i
  u(\bm{x}_i)^\top v(\bm{x}_i) - \overline{\varphi^\star}_u
\end{array}
\right] \label{defxpu}\:\:,\\
\bm{\mu}^+_v \hspace{-0.40cm} & \defeq & \hspace{-0.40cm} \left[
\begin{array}{c}
\hspace{-0.2cm} v(\bm{\mu}) \hspace{-0.2cm}\\
\hspace{-0.2cm}\varphi(v(\bm{\mu})) \hspace{-0.2cm}
\end{array}\hspace{-0.2cm}
\right] \hspace{-0.1cm}= \hspace{-0.1cm}\left[
\begin{array}{c}
v(\bm{\mu})\\
u(\bm{\mu})^\top v(\bm{\mu}) - \varphi^\star(u(\bm{\mu}))
\end{array}
\right] \label{defmupu}\:\:.
\end{eqnarray}
\end{corollary}
(Proof omitted) Finally, we provide a general characterization of the population
minimizers for a general $g = f \circ u$. This is a generalization of
the orthogonality property in (\ref{orthmuu2}), which is interesting
since $\bm{\delta}^+$ is formulated in the $v$
coordinate mapping while $\bm{z}^+$ is formulated in the $u$
coordinate mapping.

\begin{theorem}\label{zigenphi}
Let $\struct{u}{v}{\varphi}$ be a geometric structure. Suppose $g(\bm{x}) = f(u(\bm{x}))$, with $f$ differentiable. For any ${\mathcal{S}}$
and any $\bm{\mu}, \bm{z} \in {\mathbb{R}}^d$, define $\bm{\delta}^+_v,
\bm{z}^+_u \in {\mathbb{R}}^{d+1}$ with:
\begin{eqnarray}
\bm{\delta}^+_v \hspace{-0.3cm} & \defeq &  \hspace{-0.3cm} \left[
\begin{array}{c}
\hspace{-0.2cm}\overline{\bm{x}}_v-v(\bm{\mu}) \hspace{-0.2cm}\\
\hspace{-0.2cm}\overline{\varphi}_v - \varphi(v(\bm{\mu})) \hspace{-0.2cm}
\end{array}\hspace{-0.2cm}
\right] = \underbrace{\left[
\hspace{-0.2cm}\begin{array}{c}
\overline{\bm{x}}_v\\
\overline{\varphi}_v
\end{array}\hspace{-0.2cm}
\right]}_{\bm{x}^+_v} - \underbrace{\left[
\hspace{-0.2cm}\begin{array}{c}
v(\bm{\mu})\\
\varphi(v(\bm{\mu}))
\end{array}\hspace{-0.2cm}
\right]}_{\bm{\mu}^+_v}\:\:, \label{defdeltav}\\
\bm{z}^+_u \hspace{-0.3cm} & \defeq & \hspace{-0.3cm} \left[
\hspace{-0.2cm}\begin{array}{c}
f(u(\bm{\mu})) \times \bm{z} + \nabla f(u(\bm{\mu}))^\top
\bm{z} \times u(\bm{\mu}) \\
- \nabla f(u(\bm{\mu}))^\top \bm{z}
\end{array}\hspace{-0.2cm}
\right] \:\:.\label{defzv}
\end{eqnarray}
Then any right population minimizer $\bm{\mu}$ for the $v$-conformal
divergence $D^v_{\varphi,g}$ satisfies $\bm{\delta}^+_v \bot
\bm{z}^+_u$, for any valid direction $\bm{z}$.
\end{theorem}
(Proof omitted)

\section{Robustness of the population minimizers}\label{sro}


Suppose we add an outlier element $\bm{x}_*$ with small weight $0<\epsilon<1$
to $\mathcal{S}$. The population minimizer (left or right) of $\mathcal{S}$,
$\bm{\mu}$, eventually drifts to a new population minimizer
$\bm{\mu}_* =  \bm{\mu} + \epsilon \bm{\delta}_{\bm{\mu}}$ of ${\mathcal{S}} \cup
\{\bm{x}_*\}$. $\bm{\delta}_{\bm{\mu}}$ is called the influence function of
$\bm{x}_*$ \cite{lTB}. A population minimizer is robust to outliers iff the magnitude of
$\bm{\delta}_{\bm{\mu}}$ is bounded, as explained in the following
definition where $0<\tau<1$ is any small constant.

\begin{definition}\label{drob}
The population minimizer of some divergence $D$ is 
robust to outliers when, for any outlier $\bm{x}_*$ and any weight $0<\epsilon<1-\tau$, $\|\bm{\delta}_{\bm{\mu}}\|_2 \leq C$,
where $C$ does not depend
upon $\bm{x}_*$ nor $\epsilon$.
\end{definition}

Robustness according to Definition
\ref{drob} is stronger than in the model of
\cite{lTB,vlanTB} as our robustness strictly implies theirs (which
relies on very small weights $\epsilon$). So the Lemma to follow is
a twofolds generalization of the results of \cite{lTB,vlanTB}, not only
from the standpoint of the divergences, but also from the model's.

\begin{lemma}\label{robl}
Let $\struct{u}{v}{\varphi}$ be a geometric structure.
Suppose the following assumptions are verified: (i) $g(\bm{x}) =
O(1), \forall \bm{x}$, (ii) $\|u(\bm{x})\|_2
= O(1/g(\bm{x})), \forall {\bm{x}}$, (iii) the minimal eigenvalue of
$J_u^\top J_u$ is $\lambda>0$, $J_u$ being the Jacobian of $u$. Then under assumptions (i-iii), the left population minimizer of $v$-conformal divergence $D^v_{\varphi, g}$ is robust
to outliers.
\end{lemma}
(proof in Subsection \ref{sec-robl}) Lemma \ref{robl} generalizes the robustness of the left population
centers of total Bregman divergences (Theorem III.2 in \cite{vlanTB}),
for which $g=K \gbot, v = \mathrm{Id}, u = \nabla\varphi$ (the Jacobian
of $u$ being the Hessian of $\varphi$, it satisfies assumption (iii)
since $\varphi$ is strictly convex). 

The right population minimizer is unfortunately not robust to outliers
for any $g$ according to Definition \ref{drob}, yet it satisfies in a general setting of $v$-conformal divergences, a
weaker notion of robustness which says that the influence function
must be properly bounded by a divergence between
$\bm{x}_*$ and $\bm{\mu}$, as long as $\bm{x}_*$ does not deviate too
much from $\bm{\mu}$ in the $v$-coordinate mapping. This last notion exploits the fact that convex
function are locally Lipschitz.
\begin{definition}\label{dwrob}
Let $\struct{u}{v}{\varphi}$ be a geometric structure. 
The population minimizer of some $v$-conformal divergence
$D^v_{\varphi, g}$ is 
$K$-weakly robust to outliers when for any outlier $\bm{x}_*$ and any
weight $0<\epsilon<1$:
\begin{eqnarray}
\lefteqn{|\varphi(v(\bm{x}_*))-\varphi(v(\bm{\mu}))| \leq L \|v(\bm{x}_*) -
  v(\bm{\mu})\|_2}\nonumber\\
 & \Rightarrow \|\bm{\delta}_{\bm{\mu}}\|_2 \leq K \ell(L) \|\bm{x}_* -
  \bm{\mu}\|_2\:\:, 
\end{eqnarray}
where $K\geq 0$ is not a function of $\bm{x}_*$ or $\epsilon$, and
$\ell(L)$ is a linear function in $L$.
\end{definition}

We now show that the right population minimizer is $K$-weakly robust to
outliers, for a $K$ which depends solely on the coordinate mapping $v$.
We assume in the Lemma that $\bm{0} \in
\mathrm{im} u$, which is a mild assumption as it postulates
in the $\struct{u}{v}{\varphi}$-geometric structure that the gradient
$\nabla_\varphi$ has a root in coordinate mapping $v$. We exploit the
fact that any matrix $A \in
{\mathbb{R}}^{d\times d}$ satisfies $A^\top A \succeq 0$, where
``$\succeq$'' means positive semi-definite.
\begin{lemma}\label{robr}
Let $\struct{u}{v}{\varphi}$ be a geometric structure, and let $f \defeq g\circ
u^{-1}$. We make the following assumptions: (i) $\bm{0} \in
\mathrm{im} u$, (ii) $f(\bm{z}) \neq 0, \forall
\bm{z}$, (iii) the ratio
of the maximal to the minimal eigenvalue of $J_v^\top J_v$, noted
$\lambda_v$, is finite, where $J_v$
is the Jacobian of $v$. Then the right population
minimizer of $v$-conformal divergence $D^v_{\varphi, g}$ is $\sqrt{\lambda_v}$-weakly
robust to outliers.
\end{lemma}
(Proof in Subsection \ref{subrobr})


\section{Discussion}\label{sdis}

In this Section, we discuss several aspects of population minimizers
in the setting of conformal divergences; in particular, we discuss further the geometric structure relation,
the approximation of the right population minimizers in the 1D
setting, the existence of symmetric conformal divergences,
and the uniqueness of the right population minimizer. 

\noindent\hspace{0.5cm}\textbf{The nature of the $\struct{u}{v}{\varphi}$-geometric
  structure relation} ---
The $\struct{u}{v}{\varphi}$-geometric structure has been introduced in
the context of information geometry to provide a way to compute and
analyze the dually flat $(\bm{\eta},\bm{\theta})$ coordinate system
arising \textit{e.g.} in exponential families and ordinary Bregman divergences,
through a single source parameter which is originally a distribution
\cite{aTQ}. To state the key result about the
$\struct{u}{v}{\varphi}$-geometric structure, we consider two strictly
monotonous differentiable functions $u(\xi)$ and $v(\xi)$ with $u(0)=v(0)=0$. 
Consider the positive measures on $\mathbb{R}^{d+1}_+$, and denote by
$m(x,\xi)=\sum_{i=1}^{d+1} \xi_i {{1\!\!1}}_{x=x_i}$ a positive distribution
computed from $\mathcal{S}=\{x_1, x_2, ..., x_{d+1}\}$, where
${1\!\!1}$ is the indicator variable. $\bm{\xi} \defeq [\xi_1 \:\:
\xi_2 \:\:...\:\: \xi_{d+1}]^\top$ defines a coordinate system from
which we may define two coordinate systems $\bm{\eta},\bm{\theta}$ of ${\mathbb{R}}^{d+1}$
with $\theta^i \defeq u(\xi_i)$ and $\eta^i \defeq v(\xi_i)$.
These
coordinate systems have the following interesting
information-geometric properties.

\begin{theorem}\label{thig}\cite{aTQ}
The $\struct{u}{v}{\varphi}$-geometric structure is dually flat, with the following two potential functions:
\begin{eqnarray*}
\psi(\bm{\theta}) & \defeq & \sum \int \int
\frac{v'(u^{-1}(\theta^i))}{u'(u^{-1}(\theta^i))} (\dtheta^i)^2\:\:,\\
\varphi(\bm{\eta}) & \defeq & \sum \int \int
\frac{u'(v^{-1}(\eta^i))}{v'(v^{-1}(\eta^i))} (\deta^i)^2\:\:;
\end{eqnarray*}
the divergence between two $\bm{p}$ and $\bm{q}$ is given by:
\begin{eqnarray*}
D(\bm{p};\bm{q}) & \defeq &
\psi(\bm{\theta}_{\bm{p}})+\varphi(\bm{\eta}_{\bm{q}})-
\bm{\theta}_{\bm{p}} \cdot \bm{\eta}_{\bm{q}} \:\:,
\end{eqnarray*}
and the metric in the $\bm{\theta}$ coordinate system is:
\begin{eqnarray*}
g_{ij}(\bm{\theta}) & \defeq & \frac{v'(\xi_i)}{u'(\xi_i)}\delta_{ij} \:\:.
\end{eqnarray*}
\end{theorem}
One may check that $u \circ v^{-1}$ defines $\nabla \varphi$,
and that $D$ is an ordinary Bregman divergence. One important example
is Amari's $(\alpha, \beta)$ structure ($\alpha, \beta > 0$) for which $u(\xi_i) \defeq
\xi_i^\alpha$ and $v(\xi_i) \defeq \xi_i^\beta$, which helps to see the
usefulness of the $\struct{u}{v}{\varphi}$-geometric structure in the
context of clustering: assuming $\bm{\xi}$ is a source parameter 
recorded in data, one can jointly tune $\alpha, \beta$ to
tune the coordinate system of the divergence without changing its
generator $\varphi$ as long as $\alpha/\beta$ remains a fixed
constant (because $\varphi(\bm{\eta}) = (1+\alpha/\beta)^{-1} \sum_i
\eta_i^{1+\alpha/\beta}$, omitting the additive constant which does
not change the divergence). Thus, we get new free parameters to tune that adapt the
coordinate system from which the divergence is computed, which we may
use to get improved
clustering results.

We now show that, if we accept to change the generator, then we may
have a significant freedom in picking and changing the coordinate
mappings $u$ and $v$. We study the nature of the $\struct{u}{v}{\varphi}$-geometric structure, and define a tolerance relation \cite{sTST} as a binary relation
which is reflexive and symmetric but not necessarily transitive. An equivalence
relation is reflexive, symmetric and transitive. We consider the 
``geometric structure'' binary
relation, $(u,v)$ (without reference to $\varphi$), which holds when there exists some $\varphi$
such that $\struct{u}{v}{\varphi}$ is a geometric structure.
\begin{lemma}\label{luv}
The ``geometric structure'' relation is a tolerance relation. It is an
equivalence relation in the subset of functions ${\mathcal{S}}_\phi$
indexed by some strictly convex differentiable $\phi$ and defined by: ${\mathcal{S}}_\phi \defeq \{\varphi :
\mathrm{H}\varphi(\nabla \phi) = \mathrm{H}\phi P_\phi
D_\phi^{-1} D P_\phi^\top\}$, where $P_\phi, D_\phi$ are the eigenspace and
eigenvalues matrix of $\mathrm{H}\phi$ and $D\succ 0$ is diagonal.
\end{lemma}
(Proof in Appendix, Subsection \ref{proof_luv})
Hence, for example, the geometric structure relation is an equivalence relation on
any subset of positive definite quadratic forms that have the same
eigenspace. The compactness and convexity of some of these subgroups ${\mathcal{S}}_\phi$
may be interesting from the clustering standpoint to learn the
$\struct{u}{v}{\varphi}$-geometric structure (see Section \ref{srpm}).

\noindent\hspace{0.5cm}\textbf{Simple right population minimizers} --- The following corollary
is a safe-check of Lemma \ref{ld} which states when the right
population minimizer has simple forms.
\begin{corollary}\label{csrpm}
Suppose ${\mathcal{S}}$ contains at least two distinct elements. The right population minimizer of $D_{\varphi, g}$ on set
${\mathcal{S}}$ is:
\begin{enumerate}
\item always the arithmetic average
  (\textit{i.e.} $\bm{\mu} = \overline{\bm{x}}$) iff $g(\bm{y})$ is
  constant;
\item always the $\varphi$-mean (\textit{i.e.}
  $\varphi(\bm{\mu}) = \overline{\varphi}$)
  iff $\varphi(\bm{x}) = K\int 1/h(\bm{u}^\top
\bm{x}) + K'$, with (i) $K, K'$ and vector $\bm{u}$ constants, (ii) $g(\bm{x}) =
h(\bm{u}^\top \bm{x})$ for some function $h:{\mathbb{R}} \rightarrow {\mathbb{R}}$ strictly monotonous with derivative sign opposite to that
  of $K$. 
\end{enumerate}
\end{corollary}
(proof in Subsection \ref{sec-csrpm})

\noindent\hspace{0.5cm}\textbf{Only one symmetric conformal divergence} ---
We show that there exists a single 1D symmetric conformal divergence in the
$\struct{u}{v}{\varphi}$-geometric structure, the square loss, $D_{\varphi,
  g}(v(x):v(y)) \propto (v(x) - v(y))^2$. As a corollary, it shows
that there is no symmetric total Bregman divergence.
The proof is made in the 1D case, that is, when the domain and image
of $u$ and $v$ is ${\mathbb{R}}$, and it can be extended at no cost to $d$D
separable conformal divergences, for which $g(u(\bm{y}))
D_\varphi(v(\bm{x}) : v(\bm{y})) \defeq \sum_i g(u(y^i)) D_\varphi(v(x^i) : v(y^i))$.
\begin{lemma}\label{lsd}
Let $\struct{u}{v}{\varphi}$ be a geometric structure and $D_{\varphi, g}$ a
conformal divergence for some strictly convex twice differentiable $\varphi$. Suppose that $\forall x, y$:
\begin{eqnarray}
\lefteqn{g(u(y)) D_\varphi(v(x) : v(y))}\nonumber\\ 
 & = & g(u(x)) D_\varphi(v(y) :
v(x))\:\:.\label{eq00}
\end{eqnarray} 
Then \textbf{(i)} $g(.) = K_1$, \textbf{(ii)}
$v=\ell(u)$, \textbf{(iii)} $\varphi(x) = K_2x^2+ \ell(x)$ for some
constants $K_1 > 0, K_2 > 0$, where $\ell(.)$ is a linear function in
its argument.
\end{lemma}
(proof in Subsection \ref{sec-lsd}) Thus, conformal divergence are not metrics, yet they can be used to
craft metrics. To ensure that symmetry and triangle inequality are
met without violating non-negativity nor the identity of
indiscernibles, we can search for the $\alpha \in (0,1]$ with which $(g(u(\bm{y}))
D_\varphi(v(\bm{x}) : v(\bm{y})) + g(u(\bm{x}))
D_\varphi(v(\bm{y}) : v(\bm{x})))^\alpha$ meets the triangle
inequality, or use \cite{abbBD}'s method.

\noindent\hspace{0.5cm}\textbf{Fast approximation of right population minimizers}
--- We now show that under mild assumptions on $\varphi$, candidates for
right population minimizer may be easily located and approximated in the 1D setting. Assume wlog that 
${\mathcal{S}}$ is ordered, that is $x_1\leq x_2\leq ... \leq x_n$. Whenever $\varphi$ is bijective over $[x_1, x_n]$, we
define the $\varphi$-mean:
\begin{eqnarray}
\overline{x}_\varphi & \defeq & \varphi^{-1}(\overline{\varphi})\:\:.\nonumber
\end{eqnarray}
Let us denote a \textit{candidate} right population minimizer as a real which is
solution of (\ref{defcorr33}). Candidate population minimizers are
critical points for the right parameter of the average divergence \cite{cOA}.
\begin{lemma}\label{lemcs}
Suppose $g(y) = K \gbot(y)$, and assume that $\varphi'$ has constant sign on $[x_1, x_n]$. Then there
exists a candidate right population minimizer $\mu$ in $[x_1, x_n]$. Furthermore, $\mu
\in [\overline{x}_\varphi,\overline{x}]$ if
$\mathrm{sign} = -$, and $\mu
\in [\overline{x}, \overline{x}_\varphi]$ if
$\mathrm{sign} = +$. Here, ``$\mathrm{sign}$'' denotes the sign of
$\varphi'$ over $[x_1, x_n]$.
\end{lemma}
(Proof in Appendix, Subsection \ref{proof_lemcs})
\begin{table*}[t]
\center
\begin{tabular}{cccc|ll}\hline\hline
 $\varphi(x)$ & Name or expression & 
 $\overline{x}_\varphi$ & Name of & Location\\ 
 & for $D_{\varphi,\gbot} (x : y)$ & & $\varphi$-mean & \\ \hline \hline
$-\log(x)$ & Total Itakura-Saito & $\prod_{i}
  {x_i^{\frac{1}{n}}}$ & Geometric mean &
$[\overline{x}_\varphi,\overline{x}]$\\ \hline
$1/x$ & $\frac{1}{x} - \frac{2}{y} + \frac{x}{y^2}$ &
$n/\sum_i{x^{-1}_i}$ & Harmonic mean & $[\overline{x}_\varphi,\overline{x}]$ \\\hline
$x^2$ & Total square loss &
$\pm \sqrt{\frac{1}{n}\sum_i {x^2_i}}$ & $\pm$ Root mean
  square &
$[\overline{x},\overline{x}_\varphi]$ or $[\overline{x}_\varphi,\overline{x}]$\\ \hline
$x^p, p\geq 2$ & Total power loss &
$\left(\sum_i {x^p_i}\right)^{\frac{1}{p}}$ & Power mean &

$[\overline{x},\overline{x}_\varphi]$\\ \hline
$\exp(x)$ & Total exp divergence& $\log\sum_i\exp
x_i$ & None & $[\overline{x}_\varphi,\overline{x}]$\\ \hline
$x\log(x)$ & Total KL & $\frac{\frac{1}{n}\sum_i {x_i
    \log x_i}}{W\left(\frac{1}{n}\sum_i {x_i
    \log x_i}\right)}$ & None & 
$[\overline{x},\overline{x}_\varphi]$\\\hline
\multirow{2}{*}{$W^{-1}(x)$} & $x(\exp(x)-\exp(y))$ & \multirow{2}{*}{$W(\frac{1}{n}\sum_i
W^{-1}(x_i))$} & \multirow{2}{*}{None} & \multirow{2}{*}{$[\overline{x},\overline{x}_\varphi]$}\\
 & $+y(y-x)\exp(y)$ & & & & \\ \hline\hline
\end{tabular}
\caption{Examples of total Bregman divergences $D_{\varphi,\gbot} (x
  : y)$, and
Location of the candidate population minimizer according to Lemma \ref{lemcs}. $W$ is Lambert $W$ function and $W^{-1}(x)$ is shorthand
  for $x \exp (x)$.}\label{overphit}
\end{table*}
Table \ref{overphit} presents some applications of Lemma
\ref{lemcs} (the domain considered for $\varphi(x) = 1/x$ is
  ${\mathbb{R}}_{+*}$). Approximating the candidate right population minimizer
$\mu$ in
the interval may be done by fitting the roots of equations of the form
$f(\mu,\overline{x},\overline{x}_\varphi) = 0$, some of
which are given below as examples:
\begin{eqnarray*}
\log(\overline{x}_\varphi) - \log(\mu) + \mu^2 - \overline{x}\mu
& = & 0 
\end{eqnarray*}
for total Itakura Saito divergence, 
\begin{eqnarray*}
2\mu^3-(2
\overline{x}^2_\varphi-1)\mu - \overline{x} & = & 0 
\end{eqnarray*}
for total square loss divergence (notice that a closed-form expression
for $\mu$ is available), 
\begin{eqnarray*}
p\mu^{2p-1} - p
\overline{x}^p_\varphi\mu^{p-1} + \mu -\overline{x} & = &
0 
\end{eqnarray*}
for total power loss divergence, 
\begin{eqnarray*}
\exp(2\mu) - \exp(\overline{x}_\varphi + \mu) + \mu -
\overline{x} & = & 0 
\end{eqnarray*}
for total exp divergence, and finally
\begin{eqnarray*}
\left(\mu \log \mu -
  \frac{\overline{x}_\varphi}{W(\overline{x}_\varphi)} \log
  \frac{\overline{x}_\varphi}{W(\overline{x}_\varphi)}\right)(1+\log\mu)\nonumber\\
+\mu - \overline{x} & = & 0\label{eqtkl}
\end{eqnarray*}
for total KL divergence.

\begin{figure}[t]
\centering
\begin{tabular}{c} 
\epsfig{file=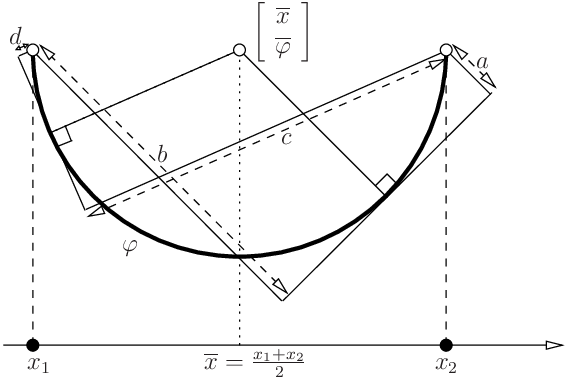, width=0.45\textwidth}
\end{tabular}
\caption{$\varphi$ is half-circle and ${\mathcal{S}} = \{x_1,
  x_2\}$. In this case, all points in $[x_1, x_2]$ are right
  population minimizers for $D_{\varphi, \gbot}$ (the total Bregman divergence equals $a+b = c+d$).}\label{f-allmu}
\end{figure}

\noindent\hspace{0.5cm}\textbf{Non-uniqueness and existence of the right population
  minimizers} ---
The left population minimizer of any $v$-conformal divergence is
unique (Lemma \ref{lleft1}). This is not always the case for the right
population minimizer. In very seldom but typical pathological cases, the
population minimizers may even span the complete domain of $\varphi$, as
displayed in Figure \ref{f-allmu}. 

We also notice that the compactness
of $\mathrm{dom}(D_{\varphi,g})$ appears necessary for the right
population minimizers to exist, as otherwise one may build
pathological Cauchy
sequences for the right divergence parameter that converge to a right
population minimizer not in $\mathrm{dom}(D_{\varphi,g})$.

\noindent\hspace{0.5cm}\textbf{Extension to scaled Bregman divergences} ---
A new generalization of ordinary Bregman divergences has been recently
coined \cite{svOB}, called scaled Bregman divergences. A scaled
Bregman divergence is a particular case, for $g=v=\mathrm{Id}$, of what we call a scaled
conformal divergence, defined as:
\begin{eqnarray}
D^v_{\varphi,g} (x : y ; w) & \defeq & w D^v_{\varphi, g}
(x/w : y/w)\:\:,\label{defscaled}
\end{eqnarray}
for $w>0$. A conformal divergence is obtained when $w=1$. Scaled Bregman
divergences generalize other important classes of divergences such as Csisz\'ar's
$f$-divergences, and they yield explicit formulas for exponential
families for scaled Bregman power
divergences, which means they have a
significant potential for applications in clustering \cite{svOB}. It
is thus important to characterize their population minimizers. Though
it is out of the scope of our paper to extent further our results to
scaled divergences, we can give some insights into the similarities
and differences with the case $w=1$. Population minimizers are now sought with respect to some sets
$\mathcal{S} = \{x_1, x_2, ..., x_n\}$ and $\mathcal{W} = \{w_1, w_2,
..., w_n\}$, such that a left population minimizer of the ordered pair
$(\mathcal{S}, \mathcal{W})$ for $D^v_{\varphi,g}$ is defined as $\mu$
that minimizes $\sum_i D^v_{\varphi,g}(\mu:x_i;w_i)$. 
The following Lemma shows that, despite the left
population minimizer is not always available in closed form in
general (unlike $v$-conformal divergences), it is in between the
minimal and maximal values of ${\mathcal{S}}$ (like $v$-conformal
divergences). 

\begin{lemma}\label{lemsbd}
The left population minimizer of $D^v_{\varphi,g}$
over $(\mathcal{S}, \mathcal{W})$ is unique and in
$[\min_i x_i, \max_i x_i]$.
\end{lemma}
(Proof in Appendix, Subsection \ref{proof_lemsbd})
This Lemma can be extended to separable divergences in ${\mathbb{R}}^d$, to show that the
left population minimizer of scaled conformal divergences lies in $\prod_j [\min_i x^j_i, \max_i x^j_i]$.

\section{Conclusion}

We have studied the left and right population minimizers of conformal
divergences, a superset of ordinary Bregman divergences and total Bregman
divergences, in the $\struct{u}{v}{\varphi}$-geometric structure \cite{aTQ,aIG,zDF}, which generalizes
dually flat affine connections. We have characterized analytically and
geometrically the population minimizers,
shown the exhaustivity property of conformal divergences for the left
population minimizer, and the exhaustivity of total Bregman
divergences for the right population minimizers. We do believe that these results, as
well as additional results we provide
on
the robustness of the population minimizers, the nature of the $\struct{u}{v}{\varphi}$
geometric structure
relation,
and the simple approximation of 1D population minimizers, 
shall be useful to widen the scope of existing clustering algorithms
and/or develop algorithmically new clustering algorithms 
relying on broad classes of distortions that escape the conventional framework of
ordinary Bregman divergences, as \textit{e.g.} recently
initiated with total Bregman divergences or scaled Bregman divergences.

\section{Acknowledgments}

The authors thank Meizhu Liu, and the reviewers for
insightful comments that helped to significantly improve the
paper. 

\bibliographystyle{IEEEtran}
\bibliography{bibgen}

\begin{thebibliography}{10}
\providecommand{\url}[1]{#1}
\csname url@samestyle\endcsname
\providecommand{\newblock}{\relax}
\providecommand{\bibinfo}[2]{#2}
\providecommand{\BIBentrySTDinterwordspacing}{\spaceskip=0pt\relax}
\providecommand{\BIBentryALTinterwordstretchfactor}{4}
\providecommand{\BIBentryALTinterwordspacing}{\spaceskip=\fontdimen2\font plus
\BIBentryALTinterwordstretchfactor\fontdimen3\font minus
  \fontdimen4\font\relax}
\providecommand{\BIBforeignlanguage}[2]{{%
\expandafter\ifx\csname l@#1\endcsname\relax
\typeout{** WARNING: IEEEtran.bst: No hyphenation pattern has been}%
\typeout{** loaded for the language `#1'. Using the pattern for}%
\typeout{** the default language instead.}%
\else
\language=\csname l@#1\endcsname
\fi
#2}}
\providecommand{\BIBdecl}{\relax}
\BIBdecl

\bibitem{vwgCS}
U.~von Luxburg, R.-C. Williamson, and I.~Guyon, ``Clustering: science or art?''
  in \emph{Workshop on Unsupervised and Transfer Learning}, 2012, pp. 65--79.

\bibitem{dlrML}
A.~P. Dempster, N.~M. Laird, and D.~B. Rubin, ``Maximum likelihood from
  incomplete data via the {EM} algorithm,'' \emph{J. of the Royal Stat. Soc.
  B}, vol.~39, pp. 1--38, 1977.

\bibitem{mqSM}
J.~McQueen, ``Some methods for classification and analysis of multivariate
  observations,'' in \emph{Proc.\ of the 5$^{th}$ Berkeley symposium on
  mathematical statistics and probability}, 1967, pp. 281--297.

\bibitem{bmdgCW}
A.~Banerjee, S.~Merugu, I.~Dhillon, and J.~Ghosh, ``Clustering with bregman
  divergences,'' in \emph{Proc.\ of the $4^{th}$ SIAM International Conference
  on Data Mining}, 2004, pp. 234--245.

\bibitem{kSL}
A.-N. Kolmogorov, ``Sur la notion de moyenne,'' \emph{Atti della R. Accademia
  Nazionale dei Lincei}, vol.~12, pp. 388--391, 1930.

\bibitem{nUE}
M.~Nagumo, ``{\"U}ber eine klasse von mittelwerte,'' \emph{Japanese Journal of
  Mathematics}, vol.~7, pp. 71--79, 1930.

\bibitem{afAC}
J.~Abernethy and R.-M. Frongillo, ``A characterization of proper scoring rules
  for linear properties,'' in \emph{Proc.\ of the 25$^{~th}$ COLT}, 2012, pp.
  1--14.

\bibitem{bgwOT}
A.~Banerjee, X.~Guo, and H.~Wang, ``On the optimality of conditional
  expectation as a bregman predictor,'' \emph{IEEE Trans. on Information
  Theory}, vol.~51, pp. 2664--2669, 2005.

\bibitem{lTB}
M.~Liu, ``Total {B}regman divergence, a robust divergence measure, and its
  applications,'' Ph.D. dissertation, University of Florida, 2011.

\bibitem{ehllIT}
F.~Escolano, E.-R. Hancock, M.~Liu, and M.-A. Lozano, ``Information-theoretic
  dissimilarities for graphs,'' in \emph{Similarity-Based Pattern Recognition},
  ser. Lecture Notes in Computer Science, 2013, vol. 7953, pp. 90--105.

\bibitem{elhTB}
F.~Escolano, M.~Liu, and E.-R. Hancock, ``Tensor-based total bregman
  divergences between graphs,'' in \emph{{IEEE} International Conference on
  Computer Vision Workshops}, 2011, pp. 1440--1447.

\bibitem{llyyhCT}
M.~Liu, L.~Lu, X.~Ye, S.~Yu, and H.~Huang, ``Coarse-to-fine classification via
  parametric and nonparametric models for computer-aided diagnosis,'' in
  \emph{Proc.\ of the 20$^{~th}$ ACM International Conference on Information
  and Knowledge Management}, 2011, pp. 2509--2512.

\bibitem{lvRA}
M.~Liu and B.-C. Vemuri, ``Robust and efficient regularized boosting using
  total bregman divergence,'' in \emph{Proc.\ of the 24$^{th}$ IEEE CVPR},
  2011, pp. 2897--2902.

\bibitem{lvanTB}
M.~Liu, B.-C. Vemuri, S.-I. Amari, and F.~Nielsen, ``total {B}regman divergence
  and its applications to shape retrieval,'' in \emph{Proc.\ of the 23$^{rd}$
  IEEE CVPR}, 2010, pp. 3463--3468.

\bibitem{lvdAR}
M.~Liu, B.-C. Vemuri, and R.~Deriche, ``A robust variational approach for
  simultaneous smoothing and estimation of {DTI},'' \emph{NeuroImage}, vol.~67,
  pp. 33 -- 41, 2013.

\bibitem{vlanTB}
B.-C. Vemuri, M.~Liu, S.-I. Amari, and F.~Nielsen, ``Total bregman divergence
  and its applications to {DTI} analysis,'' \emph{IEEE Transactions on Medical
  Imaging}, vol.~30, no.~2, pp. 475--483, 2011.

\bibitem{lvanSR}
M.~Liu, B.~C. Vemuri, S.-I. Amari, and F.~Nielsen, ``Shape retrieval using
  hierarchical total bregman soft clustering,'' \emph{IEEE T. PAMI}, vol.~34,
  no.~12, pp. 2407--2419, 2012.

\bibitem{rglTB}
A.~R.~M. y~Ter{\'a}n, M.~Gouiff{\`e}s, and L.~Lacassagne, ``Total bregman
  divergence for multiple object tracking,'' in \emph{Proc.\ of the 20$^{~th}$
  IEEE International Conference on Image Processing}, 2013.

\bibitem{svOB}
W.~Stummer and I.~Vajda, ``On {B}regman distances and divergences of
  probability measures,'' \emph{IEEE Trans. on Information Theory}, vol.~58,
  pp. 1277--1288, 2012.

\bibitem{aTQ}
S.-I. Amari, ``New developments of information geometry (17): {T}sallis
  $q$-entropy, escort geometry, conformal geometry,'' in \emph{Mathematical
  Sciences (suurikagaku)}.\hskip 1em plus 0.5em minus 0.4em\relax Science
  Company, October 2012, no. 592, pp. 73--82, in japanese.

\bibitem{aIG}
------, ``New developments of information geometry (26): Information geometry
  of convex programming and game theory,'' in \emph{Mathematical Sciences
  (suurikagaku)}.\hskip 1em plus 0.5em minus 0.4em\relax Science Company,
  November 2013, no. 605, pp. 65--74, in japanese.

\bibitem{zDF}
J.~Zhang, ``Divergence function, duality, and convex analysis,'' \emph{Neural
  Computation}, vol.~16, pp. 159--195, 2004.

\bibitem{anMO}
S.-I. Amari and H.~Nagaoka, \emph{Methods of Information Geometry}.\hskip 1em
  plus 0.5em minus 0.4em\relax Oxford University Press, 2000.

\bibitem{bDM}
M.~Basseville, ``Divergence measures for statistical data processing,''
  \emph{Signal Processing}, vol.~93, pp. 621--633, 2013.

\bibitem{omaAD}
A.~Ohara, H.~Matsuzoe, and S.-I. Amari, ``A dually-flat structure on the space
  of escort distributions,'' \emph{Journal of Physics : conference series},
  vol. 201, no. 012012, 2010.

\bibitem{bnnBV}
J.-D. Boissonnat, F.~Nielsen, and R.~Nock, ``{Bregman Voronoi Diagrams},''
  \emph{DCG}, vol.~44, no.~2, pp. 281--307, 2010.

\bibitem{sTST}
A.-B. Sossinsky, ``Tolerance space theory and some applications,'' \emph{Acta
  Applicandae Mathematicae}, vol.~5, pp. 137--167, 1986.

\bibitem{abbBD}
S.~Acharyya, A.~Banerjee, and D.~Boley, ``Bregman divergences and triangle
  inequality,'' in \emph{Proc.\ of the $13^{th}$ SIAM International Conference
  on Data Mining}, 2013, pp. 476--484.

\bibitem{cOA}
F.-H. Clarke, \emph{Optimization and Nonsmooth Analysis}.\hskip 1em plus 0.5em
  minus 0.4em\relax Wiley, 1989.

\end{thebibliography}

\section{Appendix}\label{appen}

\subsection{Proof of Lemma \ref{lleft1}}\label{proof_lleft1}

Any left population minimizer $\bm{\mu}$ satisfies $\nabla
\sum_i{D_{\varphi, g}^v (\bm{\mu} : \bm{x}_i)} = \bm{0}$, and so, after
simplification, we obtain:
\begin{eqnarray}
J_v(\bm{\mu}) \sum_i {g(\bm{x}_i) \left(\nabla\varphi(v(\bm{\mu})) -
    \nabla\varphi(v(\bm{x}_i))\right)} & = & \bm{0}\:\:,\label{leftpm}
\end{eqnarray}
where $J_v$ is the Jacobian of $v$. Since $v$ is bijective, the null
space of $J_v$ is reduced to $\{\bm{0}\}$, and so we must have:
\begin{eqnarray*}
\sum_i {g(\bm{x}_i) \left(\nabla\varphi(v(\bm{\mu})) -
    \nabla\varphi(v(\bm{x}_i))\right)} & = & \bm{0} \:\:,
\end{eqnarray*}
which, after solving for $\bm{\mu}$, yields:
\begin{eqnarray*}
\bm{\mu} \hspace{-0.2cm}& = &  \hspace{-0.2cm} v^{-1} \nabla\varphi^{-1}\left( \frac{1}{\sum_i g(\bm{x}_i)} \sum_i g(\bm{x}_i) \nabla\varphi(v(\bm{x}_i))\right)\:\:.
\end{eqnarray*}
There remains to use the $u$-coordinate mapping to obtain (\ref{propuv}).

\subsection{Proof of Lemma \ref{lleft2}}\label{proof_lleft2}

Let $\bm{x}_u \defeq u(\bm{x}), \forall \bm{x}$ and let
${\mathcal{S}}_u \defeq \{(\bm{x}_i)_u, i = 1, 2, ..., n\}$. We have $\bm{\mu}_u =
\sum_i w_i (\bm{x}_i)_u$ and 
\begin{eqnarray}
D(\bm{\mu} : \bm{x}_i) & =& 
D(u^{-1}(\bm{\mu}_u) : u^{-1}((\bm{x}_i)_u)) \nonumber\\
 & \defeq & D_2(\bm{\mu}_u :
(\bm{x}_i)_u)\label{dd0}\:\:.
\end{eqnarray} 
From (\ref{dd0}), the assumptions on $D$ and the
properties of $u$, it follows that $D_2$ is
non-negative, differentiable, satisfies $D_2(\bm{x} : \bm{x}) = 0,
\forall \bm{x}$, and its left population minimizer over ${\mathcal{S}}_u$
is the weighted arithmetic average $\bm{\mu}_u$: it is thus an ordinary Bregman
divergence \cite{afAC} with:
\begin{eqnarray}
D_2(\bm{\mu}_u :
(\bm{x}_i)_u) & = & w_i D_\phi ((\bm{x}_i)_u : \bm{\mu}_u)\:\:,\label{dd1}
\end{eqnarray}
for some strictly convex differentiable $\phi$. Calling to convex
conjugates, we obtain:
\begin{eqnarray}
\lefteqn{D_\phi ((\bm{x}_i)_u : \bm{\mu}_u)}\nonumber\\
 & = & D_{\phi^\star}
(\nabla\phi(\bm{\mu}_u) : \nabla\phi((\bm{x}_i)_u))\nonumber\\
 & = & D_{\phi^\star}
( (\nabla\phi \circ u)(\bm{\mu}) : (\nabla\phi\circ
u)(\bm{x}_i))\nonumber\\
 & \defeq & D_\varphi(v(\bm{\mu}) : v(\bm{x}_i))\label{dd2}\:\:,
\end{eqnarray}
with $\varphi \defeq \phi^\star$ and $\struct{v}{u}{\phi}$-geometric structure. $\struct{u}{v}{\varphi}$ is thus a geometric structure and merging
(\ref{dd0} --- \ref{dd2}), we obtain:
\begin{eqnarray}
D(\bm{\mu} : \bm{x}_i) & = & w_i D_\varphi(v(\bm{\mu}) :
v(\bm{x}_i))\:\:.\nonumber\\
 & = & g(\bm{x}_i) D_\varphi(v(\bm{\mu}) :
v(\bm{x}_i))\:\:, \label{dd3}
\end{eqnarray}
for some $g$ admitting continuous directional derivatives which meets $g(\bm{x}_i) = w_i,
\forall i = 1, 2, ..., n$ (we can pick \textit{e.g.} a degree-$n$ polynomial). We obtain
(\ref{dd4}), as claimed. This ends the proof of Lemma \ref{lleft2}.

\subsection{Proof of Lemma \ref{ld}}\label{proof_ld}

The first part of the proof is standard, and shows that
$\proj_{{\mathcal{S}},\varphi,g}$ is the set of critical points for the right parameter \cite{cOA}. Assume $\bm{\mu}$ is a population
minimizer, and define $\mathrm{D}_{t,\bm{z}} D_{\varphi,g}({\mathcal{S}}:\bm{y}) \defeq
\sum_i {\mathrm{D}_{t,\bm{z}} D_{\varphi,g}(\bm{x}_i:\bm{y})}$.
Fix any valid direction $\bm{z}$. Because $\bm{\mu}$ is a right population
minimizer, it comes $\mathrm{D}_{t,\bm{z}} D_{\varphi,g}({\mathcal{S}}:\bm{\mu})\leq
0$ for $t\leq 0$, and $\mathrm{D}_{t,\bm{z}} D_{\varphi,g}({\mathcal{S}}:\bm{\mu}) \geq
0$ for $t\geq 0$. Since directional derivatives are defined in
direction $\bm{z}$, we obtain $0\leq \lim_{t\downarrow 0}
\mathrm{D}_{t,\bm{z}} D_{\varphi,g}({\mathcal{S}}:\bm{\mu}) =
\mathrm{D}_{\bm{z}} D_{\varphi,g}({\mathcal{S}}:\bm{\mu}) = \lim_{t\uparrow 0}
\mathrm{D}_{t,\bm{z}} D_{\varphi,g}({\mathcal{S}}:\bm{\mu})\leq 0$ and
$\bm{\mu}$ is a solution of:
\begin{eqnarray}
\lim_{t\rightarrow 0}
\mathrm{D}_{t,\bm{z}} D_{\varphi,g}({\mathcal{S}}:\bm{\mu})& = &
\mathrm{D}_{\bm{z}} \sum_i D_{\varphi,g}(\bm{x}_i:\bm{\mu}) \nonumber\\
 & = & 0\:\:. \label{defder2}
\end{eqnarray}
We plug in (\ref{defder2}) the expression of $D_{\varphi,g}$ and obtain
that for any right population minimizer
$\bm{\mu}$, the following holds:
\begin{eqnarray}
\lefteqn{\frac{1}{n} \mathrm{D}_{\bm{z}} \sum_i
D_{\varphi,g}(\bm{x}_i:\bm{\mu})}\nonumber\\
 & = &  \mathrm{D}_{\bm{z}}
g(\bm{\mu}) \times \frac{1}{n} \sum_i
D_\varphi(\bm{x}_i:\bm{\mu})\nonumber\\
 & & +g(\bm{\mu})
(\bm{\mu}-\overline{\bm{x}})^\top \mathrm{H}\varphi(\bm{\mu}) \bm{z} \nonumber\\
 & = & 0\:\:.\label{propM}
\end{eqnarray}
Rewriting, we thus need:
\begin{eqnarray}
\lefteqn{\mathrm{D}_{\bm{z}} g(\bm{\mu}) \times (\overline{\varphi} -
\varphi(\bm{\mu}))}\nonumber\\
 & = & (\overline{\bm{x}} -
\bm{\mu})^\top(g(\bm{\mu}) \mathrm{H}\varphi(\bm{\mu}) \bm{z}
\nonumber\\
 & &+ \mathrm{D}_{\bm{z}} g(\bm{\mu})\nabla \varphi(\bm{\mu}))\label{eqfon2}\\ 
 & = & (\overline{\bm{x}} - \bm{\mu})^\top \mathrm{D}_{\bm{z}}
 (g(\bm{\mu})\nabla \varphi(\bm{\mu}))\:\:, \nonumber
\end{eqnarray}
which implies $\bm{\mu} \in
\proj_{{\mathcal{S}},\varphi,g}$. If a population minimizer does not
belong to $\proj_{{\mathcal{S}},\varphi,g}$, it is in the non
differentiable part of the boundary, that is, in $\boundary_{\varphi,g}$.
Eq. (\ref{eqfon2}) brings:
\begin{eqnarray}
\lefteqn{\mathrm{D}_{\bm{z}} g(\bm{\mu}) \times \frac{1}{n} \sum_i
D_{\varphi}(\bm{x}_i:\bm{\mu})}\nonumber\\
 & = & g(\bm{\mu}) \times (\overline{\bm{x}} -
\bm{\mu})^\top \mathrm{H}\varphi(\bm{\mu}) \bm{z}\nonumber\:\:,
\end{eqnarray}
and so:
\begin{eqnarray}
\lefteqn{\frac{1}{n} \sum_i
D_{\varphi,g}(\bm{x}_i:\bm{\mu})}\nonumber\\
 & = & \frac{(\overline{\bm{x}} -
\bm{\mu})^\top \mathrm{H}\varphi(\bm{\mu}) \bm{z}}{\mathrm{D}_{\bm{z}}
g(\bm{\mu})} g^2(\bm{\mu})\:\:, \label{valdiv}
\end{eqnarray}
a
quantity which does not depend on the direction $\bm{z} \neq
\bm{0}$. 
Fixing as direction $\bm{z} = \overline{\bm{x}} -
\bm{\mu}$ yields the statement of (\ref{valdivmaha}).

\subsection{Proof of Theorem \ref{thh0}}\label{proof_thh0}
We first need the following Lemma.

\begin{lemma}\label{lzigen}
Suppose $g(\bm{x}) = f(\nabla \varphi(\bm{x}))$, with $\varphi$ strictly convex
twice differentiable and $f$ differentiable. $\forall \bm{\mu} \in
\proj_{{\mathcal{S}},\varphi,g}$, we have:
\begin{eqnarray}
\lefteqn{\overline{\bm{x}} -
\bm{\mu}}\nonumber\\
 & = & 
\frac{1}{f(\nabla
\varphi(\bm{\mu}))}\left( \frac{1}{n} \sum_i
D_{\varphi}(\bm{x}_i:\bm{\mu}) \right) \nonumber\\
& & \cdot \nabla f(\nabla
\varphi(\bm{\mu})) \:\:.\label{pdiff}
\end{eqnarray}
\end{lemma}
\begin{proof}
The chain rule gives 
\begin{eqnarray*}
\mathrm{D}_{\bm{z}} g(\bm{\mu}) & = &
\mathrm{D}_{\bm{z}}  (f(\nabla \varphi(\bm{\mu})))\nonumber\\
 & =& \mathrm{D}_{\mathrm{D}_{\bm{z}}\nabla\varphi(\bm{\mu})}
 f(\nabla\varphi) \nonumber\\
 & =&  \mathrm{D}_{\mathrm{H}\varphi(\bm{\mu}) \bm{z}}
 f(\nabla\varphi) \nonumber\\
 & =& \bm{z}^\top (\mathrm{H}\varphi(\bm{\mu}))^\top \nabla f(\nabla \varphi(\bm{\mu})) \nonumber\\
 & =&  \bm{z}^\top \mathrm{H}\varphi(\bm{\mu}) \nabla f(\nabla
 \varphi(\bm{\mu}))\:\:,\nonumber
\end{eqnarray*}
so that (\ref{eqfon2}) becomes:
\begin{eqnarray}
\lefteqn{\bm{z}^\top \mathrm{H}\varphi(\bm{\mu}) \left((\overline{\varphi} -
\varphi(\bm{\mu})) \times \nabla f(\nabla
\varphi(\bm{\mu}))\right)}\nonumber\\
 \hspace{-0.2cm}& =
&  \hspace{-0.3cm}\bm{z}^\top \mathrm{H}\varphi(\bm{\mu}) \left(g(\bm{\mu}) \times (\overline{\bm{x}} -
\bm{\mu})\right) \nonumber\\
 & & + \bm{z}^\top \mathrm{H}\varphi(\bm{\mu})
\left(\nabla \varphi(\bm{\mu})^\top (\overline{\bm{x}} -
\bm{\mu})\times \nabla f(\nabla \varphi(\bm{\mu}))\right)\nonumber\\
 \hspace{-0.2cm}& = &  \hspace{-0.3cm}\bm{z}^\top \mathrm{H}\varphi(\bm{\mu})
 \hspace{-0.2cm}\left( 
\begin{array}{c}
g(\bm{\mu}) \times (\overline{\bm{x}} -
\bm{\mu}) \\
+ \\
 \hspace{-0.2cm}\nabla \varphi(\bm{\mu})^\top (\overline{\bm{x}} -
\bm{\mu}) \times \nabla f(\nabla \varphi(\bm{\mu}))
\end{array} \hspace{-0.2cm}\right) \hspace{-0.2cm}\label{last}\:\:.
\end{eqnarray}
Eq. (\ref{last}) is of the form $\bm{z}^\top \mathrm{H}\varphi(\bm{\mu}) \bm{a}=
\bm{z}^\top \mathrm{H}\varphi(\bm{\mu}) \bm{b}$ which implies $\bm{a} = \bm{b}$ as
otherwise picking $\bm{z} = \bm{b} - \bm{a} \neq \bm{0}$ would contradict the
positive definiteness of the Hessian $\mathrm{H}\varphi$. After
reordering, we get eq. (\ref{pdiff}).
This ends the proof of Lemma \ref{lzigen}.
\cqfd\end{proof}
(Continued proof of Theorem \ref{thh0})
Let us fix $f(\bm{x}) = f_p(\bm{x}) \defeq
1/(1+\|\bm{x}\|_p^p)^{1/p}$. We have:
\begin{eqnarray}
\nabla f(\nabla\varphi(\bm{\mu}))  \hspace{-0.2cm} \hspace{-0.1cm}& = & \hspace{-0.2cm} \hspace{-0.1cm}
-\frac{1}{(1+\|\nabla\varphi(\bm{\mu})\|_p^p)^{1+\frac{1}{p}}} \nabla_p \varphi(\bm{\mu})\:\:,
\end{eqnarray}
where $\nabla_{p} \varphi (\bm{\mu})$ is the vector whose
$j^{th}$ coordinate is $\mathrm{sign}(\nabla^j) |\nabla^j|^{p-1}$,
where $\nabla^j$ is coordinate $j$ of $\nabla\varphi(\bm{\mu})$. This
definition brings the following relationship:
\begin{eqnarray}
\nabla \varphi (\bm{\mu})^\top \nabla_{p} \varphi (\bm{\mu}) & = & \|\nabla \varphi (\bm{\mu})\|_p^p\:\:.\label{ppp1}
\end{eqnarray}
We now use (\ref{pdiff}) with $f=f_p$ and obtain:
\begin{eqnarray}
\overline{\bm{x}} -
\bm{\mu} & = & 
-\frac{1}{1+\|\nabla\varphi(\bm{\mu})\|_p^p}\left( \frac{1}{n} \sum_i
D_{\varphi}(\bm{x}_i:\bm{\mu}) \right) \nonumber\\
 & & \times\nabla_p \varphi(\bm{\mu}) \:\:.\label{pdiff22}
\end{eqnarray}
Coordinate $j$ in $\overline{\bm{x}} -
\bm{\mu}$, $(\overline{\bm{x}} -
\bm{\mu})^j$, satisfies:
\begin{eqnarray}
\lefteqn{((\overline{\bm{x}} -
\bm{\mu})^j)^{q-1}}\nonumber\\
 & = & 
 \left(\frac{-m}{1+\|\nabla\varphi(\bm{\mu})\|_p^p}\right)^{q-1} \times
(\mathrm{sign}(\nabla^j) |\nabla^j|^{p-1})^{q-1} \nonumber\\
 & = & 
-\left(\frac{m}{1+\|\nabla\varphi(\bm{\mu})\|_p^p}\right)^{q-1} \times
\nabla^j \:\:,\label{propxm}
\end{eqnarray}
where $m\defeq (1/m) \sum_i
D_{\varphi}(\bm{x}_i:\bm{\mu}) \geq 0$. Eq. (\ref{propxm}) holds because
$(p-1)(q-1)=1$ and $q = 2k \in {\mathbb{N}}$ is even. So, we may write:
\begin{eqnarray}
\lefteqn{\|\overline{\bm{x}} - \bm{\mu}\|_q^q}\nonumber\\
 & \defeq &  \hspace{-0.3cm}\sum_j
{((\overline{\bm{x}}-\bm{\mu})^j)^q}\nonumber\\
 & = &  \hspace{-0.3cm}\sum_j
{((\overline{\bm{x}}-\bm{\mu})^j)^{q-1} \times
  (\overline{\bm{x}}-\bm{\mu})^j}\nonumber\\
& = &  \hspace{-0.3cm}-\left(\frac{m}{1+\|\nabla\varphi(\bm{\mu})\|_p^p}\right)^{q-1}
 \hspace{-0.5cm}\times (\overline{\bm{x}} -
\bm{\mu} )^\top \nabla\varphi(\bm{\mu})\label{res1}\:\:.
\end{eqnarray}
We make the inner product of (\ref{pdiff22}) with
$\nabla\varphi(\bm{\mu})$ and obtain because of (\ref{ppp1}):
\begin{eqnarray}
\lefteqn{(\overline{\bm{x}} -
\bm{\mu})^\top  \nabla\varphi(\bm{\mu})}\nonumber\\
& = & 
-\frac{\|\nabla\varphi(\bm{\mu})\|_p^p}{1+\|\nabla\varphi(\bm{\mu})\|_p^p}\left( \frac{1}{n} \sum_i
D_{\varphi}(\bm{x}_i:\bm{\mu}) \right) \:\:,\nonumber\\
& = & -\alpha (\overline{\varphi} -
  \varphi(\bm{\mu})) +
\alpha (\overline{\bm{x}} -
\bm{\mu})^\top  \nabla\varphi(\bm{\mu})\:\:,
\end{eqnarray}
with $\alpha \defeq
\|\nabla\varphi(\bm{\mu})\|_p^p/(1+\|\nabla\varphi(\bm{\mu})\|_p^p)$. We
obtain $-(1-\alpha) (\overline{\bm{x}} -
\bm{\mu})^\top  \nabla\varphi(\bm{\mu}) = \alpha (\overline{\varphi} -
  \varphi(\bm{\mu}))$, that is, after adding $(1-\alpha) (\overline{\varphi} -
  \varphi(\bm{\mu}))$ on both sides:
\begin{eqnarray}
\lefteqn{\overline{\varphi} -
\varphi(\bm{\mu})}\nonumber\\
  & = & \frac{1}{1 + \|\nabla
    \varphi(\bm{\mu})\|_p^p} \left( \frac{1}{n} \sum_i
D_\varphi(\bm{x}_i:\bm{\mu}) \right) \:\:.\label{res2}
\end{eqnarray}
We finally get from (\ref{res1}) and (\ref{res2}), using the shorthand
$m\defeq (1/n) \sum_i
D_{\varphi}(\bm{x}_i:\bm{\mu})$:
\begin{eqnarray}
\lefteqn{\|\overline{\bm{x}}^+ - \bm{\mu}^+\|_q}\nonumber\\
 & = & \left( |\overline{\varphi} - \varphi(\bm{\mu})|^q +
  \|\overline{\bm{x}} - \bm{\mu}\|_q^q \right)^\frac{1}{q}\nonumber\\
 & = & \left( (\overline{\varphi} - \varphi(\bm{\mu}))^{q-1}\times (\overline{\varphi} - \varphi(\bm{\mu})) +
  \|\overline{\bm{x}} - \bm{\mu}\|_q^q \right)^\frac{1}{q}\nonumber\\
 & = & \left( 
\begin{array}{c}
   \left(\frac{m}{1 + \|\nabla
    \varphi(\bm{\mu})\|_p^p}\right)^{q-1}(\overline{\varphi} - \varphi(\bm{\mu})) \\
- \\
\left(\frac{m}{1 + \|\nabla
    \varphi(\bm{\mu})\|_p^p}\right)^{q-1} (\overline{\bm{x}} -
\bm{\mu})^\top  \nabla\varphi(\bm{\mu})
\end{array}\right)^{\frac{1}{q}}\nonumber\\
 & = & \frac{m^{\frac{1}{p}}}{(1 + \|\nabla
    \varphi(\bm{\mu})\|_p^p)^{\frac{1}{p}}} \nonumber\\
 & & \times \left(\overline{\varphi}
  - \varphi(\bm{\mu}) -  (\overline{\bm{x}} -
\bm{\mu})^\top  \nabla\varphi(\bm{\mu})\right)^\frac{1}{q}\nonumber\\
 & = & \frac{m^{\frac{1}{p}+\frac{1}{q}}}{(1 + \|\nabla
    \varphi(\bm{\mu})\|_p^p)^{\frac{1}{p}}} \nonumber\\
 & = & \frac{m}{(1 + \|\nabla
    \varphi(\bm{\mu})\|_p^p)^{\frac{1}{p}}} \nonumber\\
 & = & \frac{1}{n} \sum_i D_{\varphi, g_p}
 (\bm{x}_i:\bm{\mu}) \nonumber\\
 & = & \frac{1}{K} \times \left(\frac{1}{n} \sum_i D_{\varphi, K g_p}
 (\bm{x}_i:\bm{\mu}) \right)\:\:, \nonumber
\end{eqnarray}
which yields the statement of Theorem \ref{thh0}.

\subsection{Proof of Lemma \ref{luv}}\label{proof_luv}

Clearly, $(u,u)$ holds since $\struct{u}{u}{\varphi}$ is a geometric structure for $\varphi \defeq (1/2) \sum_i
(x^i)^2$ so the relation is reflexive. If $\struct{u}{v}{\varphi}$ is
a geometric structure, then $\struct{v}{u}{\varphi^\star}$ is a
geometric structure, so the relation is symmetric, which completes the
proof that $(u,v)$ is a tolerance relation.

Let $(u,v)_\varphi$ and $(v,w)_\phi$ be two geometric structures. We
have $u\circ w^{-1} = \nabla\varphi \circ \nabla\phi$, and so
$J_{u\circ w^{-1}} = \mathrm{H}\varphi(\nabla\phi)
\mathrm{H}\phi$, that we want to be symmetric positive definite for
the ``geometric structure'' relation to be transitive. Both $\mathrm{H}\varphi$ and $\mathrm{H}\phi$ are
symmetric positive definite. Since (i) the product of two positive
definite matrices is positive definite iff their product is normal,
and (ii) the product of two symmetric matrices is symmetric iff their
have the same eigenspace, it follows that $\mathrm{H}\varphi(\nabla\phi) \mathrm{H}\phi \succ 0$
iff we have the diagonalizations $\mathrm{H}\varphi(\nabla\phi) = P D_1 P^\top$ and $\mathrm{H}\phi
= P D_2 P^\top$, with $P$ unitary and $D_1, D_2 \succ 0$. This finishes the proof of Lemma \ref{luv}.

\subsection{Proof of Lemma \ref{lemcs}}\label{proof_lemcs}

We suppose without loss of generality that $K=1$.
The proof relies on the study in $[x_1,x_n]$ of function $\tilde{\varphi}'_\bot(x)
\defeq -1/\tilde{\varphi}'_{\mathcal{S}}(x)$ (see eq. (\ref{defyz})), which is the slope of the line
orthogonal to the segment which links $(\overline{x},
\overline{\varphi})$ to $(\overline{x},\varphi(\overline{x}))$.

Suppose $\varphi'(x)$ is $<0$ on $[x_1,x_n]$, which implies $\varphi(x_1)
\geq \overline{\varphi}$, and so $\overline{x}_\varphi \in
[x_1,\overline{x}]$, and satisfies $\varphi(\overline{x}_\varphi) = \overline{\varphi}$. 
It comes $\tilde{\varphi}'_\bot(x)\leq
0$ on $(\overline{x}_\varphi, \overline{x}]$, with $\lim_{x \downarrow \overline{x}_\varphi}
\tilde{\varphi}'_\bot(x) = -\infty$ and
$\tilde{\varphi}'_\bot(\overline{x}) = 0$. Because
$\tilde{\varphi}'_{\mathcal{S}}(x)$ is continuous, so is
$\tilde{\varphi}'_\bot(x)$ and so there must be $\mu \in (\overline{x}_\varphi,
\overline{x}]$ such that $\tilde{\varphi}'_\bot(\mu) =
\varphi'(x)$. This $\mu$ is a candidate right population minimizer.

Suppose now that $\varphi'(x)$ is $>0$ on $[x_1,x_n]$, which implies $\varphi(x_n)
\geq \overline{\varphi}$, and so $\overline{x}_\varphi \in
[\overline{x}, x_n]$, and satisfies $\varphi(\overline{x}_\varphi) = \overline{\varphi}$. This time,
$\tilde{\varphi}'_\bot(\overline{x}) = 0$ and $\lim_{x \uparrow \overline{x}_\varphi}
\tilde{\varphi}'_\bot(x) = +\infty$, so there must be $\mu \in
[\overline{x} , \overline{x}_\varphi)$ such that $\tilde{\varphi}'_\bot(\mu) =
\varphi'(x)$. This $\mu$ is a candidate right population minimizer. This ends the proof of Lemma \ref{lemcs}.

\subsection{Proof of Lemma \ref{lemsbd}}\label{proof_lemsbd}

We build upon eq. (\ref{leftpm}). Any left population minimizer is a solution of:
\begin{eqnarray}
\lefteqn{0 = \frac{\mathrm{d}}{\mathrm{d}\mu} \sum_i w_i D^v_{\varphi, g}
(\mu/w_i : x_i/w_i)}\nonumber\\
 & = & \hspace{-0.2cm}\sum_i g\parent{\frac{x_i}{w_i}}
  v'\parent{\frac{\mu}{w_i}} \parent{u\parent{\frac{\mu}{w_i}} -
    u\parent{\frac{x_i}{w_i}}} \nonumber\\
 & = & \hspace{-0.2cm}\sum_i g\parent{\frac{x_i}{w_i}}
  v'\parent{\frac{\mu}{w_i}} \parent{\frac{\mu - x_i}{w_i}}
  u'\parent{\frac{\mu_i}{w_i}}\:\:. \label{refmumu}
\end{eqnarray}
where $\mu_i \defeq \mu + \alpha_i(x_i - \mu)$ for some
$0<\alpha_i<1$. Eq (\ref{refmumu}) is obtained after $n$ Taylor
expansions of $u$. We also have $(v \circ u^{-1})' = (v'\circ u^{-1})/(u' \circ u^{-1}) \defeq
 (\varphi^\star)''$, and so, since $\varphi^\star$ is strictly convex,
 $v'(x)$ and $u'(x)$ have the same sign. 
Since $u$ is strictly monotonous, $u'$ does not change
 sign over its domain, and so 
 the product $\pi_i \defeq g\parent{x_i/w_i}
  v'\parent{\mu/w_i} 
  u'\parent{\mu_i/w_i}$ is non negative, $\forall i$. We can summarize (\ref{refmumu}) as
  $h(\mu) \defeq \sum_i \pi_i (\mu - x_i) / w_i = 0$: since all $w_i >
  0$, we get $h(\min_i x_i) \leq 0$ and $h(\max_i x_i) \geq 0$. Since each summand
  in $h$ is the
  product of continuous functions, there must be $\mu \in [\min_i x_i, \max_i x_i]$
  such that (\ref{refmumu}) holds, and since $h$ is strictly
  increasing, there is only one such point. Since $\sum_i
  D_{\varphi,g}(\mu:x_i;w_i)$ is strictly convex in $\mu$, this is the
  left population minimizer. This ends the proof of Lemma \ref{lemsbd}.

\clearpage
\newpage

\hrule\hrule

\noindent These subsections present the additional proofs not in the published paper.

\hrule\hrule

\subsection{Proof of Theorem \ref{thex}}\label{subright}

We distinguish two cases, first
  assuming that $\overline{x} \not\in
  \proj_{{\mathcal{S}},\varphi}$. As shown in Figure \ref{f-rot}, we perform a rotation of angle
$\theta$ chosen so that $(1/n)
\sum_i{\matrice{M}_\theta \bm{x}^+_i} = \matrice{M}_\theta((1/n)
\sum_i{\bm{x}^+_i}) = \matrice{M}_\theta \overline{\bm{x}}^+ =
\matrice{M}_\theta \bm{\mu}^+$, with 
\begin{eqnarray}
\matrice{M}_\theta & \defeq & \left[
\begin{array}{cc}
\cos \theta & -\sin \theta \\
\sin \theta & \cos \theta
\end{array} \right] \:\:, \nonumber\\
\bm{x}^+_i & \defeq & \left[
\begin{array}{c}
x_i \\
\varphi(x_i)
\end{array} \right] \:\:, \nonumber
\end{eqnarray}
and $\overline{\bm{x}}^+$ and $\bm{\mu}^+$ are defined in (\ref{defxp}) and (\ref{defmup}). So, the population minimizer of
${\mathcal{S}}$ after rotation is just the average, which implies,
since the distortion is invariant to rotation after assumption (iii)
and satisfies (i) and (ii), that the distortion equals an ordinary Bregman
divergence computed after rotation, $D_{\phi^\mathrm{rot}}$, for some convex
differentiable $\phi^\mathrm{rot} : {\mathbb{R}} \rightarrow
{\mathbb{R}}$ \cite{afAC,bgwOT}. We then get:
\begin{eqnarray}
\lefteqn{D(x : \mu)}\nonumber\\
 & = & D((\matrice{M}_\theta
\bm{x}^+)^1 : (\matrice{M}_\theta
\bm{\mu}^+)^1) \nonumber\\
 & = & D_{\phi^\mathrm{rot}}((\matrice{M}_\theta
\bm{x}^+)^1 : (\matrice{M}_\theta
\overline{\bm{x}}^+)^1) \nonumber\\
 & = & \phi^\mathrm{rot}((\matrice{M}_{\theta}
\bm{x}^+)^1) - \phi^\mathrm{rot}((\matrice{M}_{\theta}
\overline{\bm{x}}^+)^1) \nonumber\\
 & & - ((\matrice{M}_\theta
\bm{x}^+)^1 - (\matrice{M}_\theta
\overline{\bm{x}}^+)^1) \phi'^{\mathrm{rot}}((\matrice{M}_\theta
\overline{\bm{x}}^+)^1)\:\:\label{lastD}
\end{eqnarray}
(exponent ``1'' refers to the $x$-coordinate). Let us denote $\phi : {\mathbb{R}}^d \rightarrow {\mathbb{R}}$ the
function obtained from $\phi^\mathrm{rot}$ by rotation
$\matrice{M}_{\theta}$ of the curve. We have:
\begin{eqnarray*}
\phi^\mathrm{rot}((\matrice{M}_{\theta}
\bm{x}^+)^1) & = & (\matrice{M}_\theta
\bm{x}^+)^2 = x\sin\theta + \phi(x)\cos\theta\:\:,\\
\phi^\mathrm{rot}((\matrice{M}_{\theta}
\overline{\bm{x}}^+)^1) & = &  (\matrice{M}_\theta
\overline{\bm{x}}^+)^2 = \mu \sin\theta +
\phi(\mu)\cos\theta\:\:
\end{eqnarray*}
(exponent ``2'' refers here to the $y$-coordinate), and
\begin{eqnarray*}
\lefteqn{((\matrice{M}_\theta
\bm{x}^+)^1 - (\matrice{M}_\theta
\overline{\bm{x}}^+)^1) \phi'^{\mathrm{rot}}((\matrice{M}_\theta
\overline{\bm{x}}^+)^1)}\nonumber\\
 & = & (x \cos\theta - \phi(x)\sin\theta - \mu \cos\theta +
 \phi(\mu)\sin\theta)\\
 & & \times \frac{\sin\theta + \phi'(\mu)\cos \theta}{\cos\theta - \phi'(\mu)\sin \theta}\:\:.
\end{eqnarray*}
Eq. (\ref{lastD}) thus becomes $D(x : \mu) = Q + R$, with $Q
\defeq (\phi(x) - \phi(\mu))\cos\theta$, and:
{\footnotesize
\begin{eqnarray}
\lefteqn{R}\nonumber\\
 & \defeq & (x - \mu)\sin\theta \nonumber\\
 & & - \left\{
\begin{array}{c}
x \cos\theta - \phi({x})\sin\theta\\
- \mu \cos\theta +
 \phi({\mu})\sin\theta
\end{array}
\right\} \times \frac{\sin\theta + \phi'(\mu)\cos
  \theta}{\cos\theta - \phi'(\mu)\sin \theta}\nonumber\\
 & = & \frac{\left\{
\begin{array}{c}
x \sin\theta \cos\theta - x \phi'(\mu)\sin^2\theta \\
-\mu \sin\theta \cos\theta + \mu
\phi'(\mu)\sin^2\theta \\
-x \sin\theta \cos\theta - x \phi'(\mu) \cos^2\theta\\
+\mu \sin\theta \cos\theta + \mu \phi'(\mu)
\cos^2\theta\\
+ (\phi({x})-\phi({\mu}))\sin\theta(\sin\theta + \phi'(\mu)\cos
  \theta)
\end{array}
\right\}}{\cos\theta - \phi'(\mu)\sin \theta}\nonumber\\
 & = & \frac{\left\{
\begin{array}{c}
- (x  - \mu) \phi'(\mu)\\
+ (\phi({x})-\phi({\mu}))(1-\cos^2\theta + \phi'(\mu)\sin\theta\cos
  \theta)
\end{array}
\right\}}{\cos\theta - \phi'(\mu)\sin \theta}\nonumber\\
 & = & \frac{\phi({x})-\phi({\mu}) - (x  - \mu)
   \phi'(\mu)}{\cos\theta - \phi'(\mu)\sin
   \theta} \nonumber\\
 && - (\phi({x})-\phi({\mu}))\cos\theta\nonumber\:\:.
\end{eqnarray}
}
We thus get:
\begin{eqnarray}
D({x} : {\mu}) & = & \frac{\phi({x})-\phi({\mu}) - (x  - \mu)
   \phi'(\mu)}{\cos\theta - \phi'(\mu)\sin
   \theta} \nonumber\\
 & = & D_{\phi, g}(x:\mu)\:\:, \label{defdd}
\end{eqnarray}
with
\begin{eqnarray}
 g({\mu}) & \defeq & \frac{\sqrt{1+\varphi'({\mu})^2}}{1 +
   \phi'(\mu) \varphi'({\mu})}\:\:, \label{defgg}
\end{eqnarray}
as indeed $\cos\theta = 1/\sqrt{1+\varphi'({\mu})^2}$ and $\sin\theta
= -\varphi'(\mu)/\sqrt{1+\varphi'({\mu})^2}$. 

Eqs (\ref{defdd}) and (\ref{defgg}) are the consequences of assumptions
(i-iv). On the other hand, assumption (iv) and (\ref{defcorr2}) imply for
right population minimizer $\mu$ of set ${\mathcal{S}}$:
\begin{eqnarray}
-\varphi'(\mu) (\overline{\varphi} - \varphi(\mu)) & = & (\overline{x} - \mu)  \label{fp222}\:\:.
\end{eqnarray}
Since $\overline{x} \neq \mu$, we obtain $\varphi'(\mu) \neq 0$, so we can
replace $\varphi'(\mu)$ in (\ref{defgg}) by its expression from
(\ref{fp222}) and obtain:
\begin{eqnarray*}
g(\mu) & = & \frac{\sqrt{1+\left(\frac{\overline{x} -
        \mu}{\overline{\varphi} -
        \varphi(\mu)}\right)^2}}{1-\phi'(\mu) \times\left(\frac{\overline{x} -
        \mu}{\overline{\varphi} -
        \varphi(\mu)}\right)}\\
 & = & \frac{\sqrt{(\overline{\varphi} -
        \varphi(\mu))^2 + (\overline{x} -
        \mu)^2}}{\overline{\varphi} -
        \varphi(\mu) - (\overline{x}-\mu)\phi'(\mu)}\\
 & = & \frac{\|\overline{\bm{x}}^+ - \bm{\mu}^+\|_2}{\overline{\varphi} -
        \varphi(\mu) - (\overline{x}-\mu)\phi'(\mu)}\:\:,
\end{eqnarray*}
where $\overline{\bm{x}}^+$ and $\bm{\mu}^+$ are defined in (\ref{defxp})
and (\ref{defmup}). For any ${\mathcal{S}}$ whose right population
minimizer on $D_{\phi, g}$
is $\mu$, we get:
\begin{eqnarray}
\lefteqn{\frac{1}{n}\sum_i D({x}_i : {\mu})}\nonumber\\
 & = & \|\overline{\bm{x}}^+ - \bm{\mu}^+\|_2 \times \frac{\overline{\phi} -
        \phi(\mu) - (\overline{x}-\mu)\phi'(\mu)}{\overline{\varphi} -
        \varphi(\mu) - (\overline{x}-\mu)\phi'(\mu)}\label{find}
\end{eqnarray}
Because of assumption (iii), we want (\ref{find}) to be invariant to
rotation of the axes. Only term $\|\overline{\bm{x}}^+ -
\bm{\mu}^+\|_2$ is invariant because of assumption (iv). Both the
numerator and the denominator after the times in (\ref{find}) are not
invariant to rotation. To have their ratio invariant, it must
therefore be independent from the choice of ${\mathcal{S}}$, and thus
constant, so we have:
\begin{eqnarray*}
\lefteqn{\overline{\phi} -
        \phi(\mu) - (\overline{x}-\mu)\phi'(\mu)}\nonumber\\
 & = & K(\overline{\varphi} -
        \varphi(\mu) - (\overline{x}-\mu)\phi'(\mu))\:\:,
\end{eqnarray*}
Taking the derivative in some $x_i$ yields $\phi'(x_i) =
K\varphi'(x_i) - (K-1)x_i\phi'(\mu)$, which implies that the right
hand side is independent of $\mu$, and since $\phi'$ cannot always be
zero, we obtain $K=1$, and so:
\begin{eqnarray}
\phi & = & \varphi + \mathrm{constant}\:\:.
\end{eqnarray}
We obtain $D_{\phi, g} = D_{\varphi, g_\bot}$,
and this completes the proof when $\overline{x} \not\in
  \proj_{{\mathcal{S}},\varphi}$.

 If $\overline{x} \in
  \proj_{{\mathcal{S}},\varphi}$ is the population minimizer,
  then 
$D(x:\mu)=D(x:\overline{x})$ is an ordinary Bregman divergence \cite{afAC,bgwOT}, say $D_\phi$ for some
$\phi$ strictly convex twice differentiable. Because of assumption (iii),
it comes in this case:
\begin{eqnarray*}
D(x:y) & = & \sqrt{\frac{1+\phi'(\overline{x})^2}{1+\phi'(y)^2}}
D_\phi(x:y)\nonumber\\
 & = & \frac{K}{\sqrt{1+\phi'(y)^2}} D_\phi (x:y) = D_{\phi,
   K g_\bot} (x:y) \:\:.
\end{eqnarray*}
Because of assumption (iv) and Lemma \ref{cororth2}, $\phi' =
\varphi'$ and so $\phi = \varphi + \mathrm{constant}$. This completes the
proof in this second case, and completes the proof of Theorem \ref{thex}.

\begin{figure}[t]
\centering
\begin{tabular}{c} 
\epsfig{file=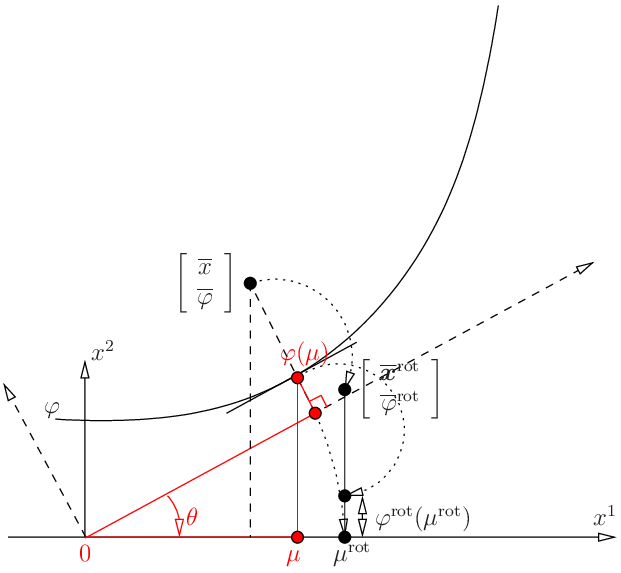, width=0.35\textwidth}
\end{tabular}
\caption{We perform a rotation of angle $\theta$ on $\varphi$ such
  that, after rotation, $\varphi'^{\mathrm{rot}}(\mu^{\mathrm{rot}}) = 0$,
  implying that $\overline{x}^{rot}$ is a population minimizer in $\proj_{{\mathcal{S}}^{\mathrm{rot}},\varphi^{\mathrm{rot}}}$.}\label{f-rot}
\end{figure}

\section{Proof of Lemma \ref{robl}}\label{sec-robl}

We use Lemma \ref{lleft1}, and
we fix ${\mathcal{S}}$ satisfying assumption (i). The left
population minimizers $\bm{\mu}$ and $\bm{\mu}_*$ satisfy:
\begin{eqnarray*}
u(\bm{\mu}) & = &  \frac{1}{\Sigma} \sum_i g(\bm{x}_i) u(\bm{x}_i)\:\:,\\
u(\bm{\mu}_*) & = &  \frac{1-\epsilon}{\Sigma_*} \sum_i g(\bm{x}_i) u(\bm{x}_i)
+ \frac{\epsilon g(\bm{x}_*)}{\Sigma_*} u(\bm{x}_*)\\
 & = & \frac{(1-\epsilon)\Sigma}{\Sigma_*} u(\bm{\mu}) + \frac{\epsilon
   g(\bm{x}_*)}{\Sigma_*} u(\bm{x}_*)\\
 & = & u(\bm{\mu}) + \underbrace{\frac{\epsilon
   g(\bm{x}_*)}{\Sigma_*}\left(u(\bm{x}_*) -
   u(\bm{\mu})\right)}_{\epsilon \bm{\delta}}\:\:,
\end{eqnarray*}
where $\Sigma \defeq \sum_i g(\bm{x}_i)$ and $\Sigma_* \defeq(1-\epsilon)\Sigma +
\epsilon g(\bm{x}_*)$. Now, a Taylor expansion of $u^{-1}$ brings:
\begin{eqnarray*}
\bm{\mu}_*  = u^{-1}\left(u(\bm{\mu}) + \epsilon\bm{\delta}\right) & = & \bm{\mu}
+\epsilon J_{u^{-1}}(\bm{\mu}_\alpha) \bm{\delta}\:\:,
\end{eqnarray*}
for some $\bm{\mu}_\alpha = u(\bm{\mu}) + \alpha \epsilon
\bm{\delta}$ with $0<\alpha <1$. We also have
$J_{u^{-1}}(\bm{\mu}_\alpha) = J^{-1}_u(u^{-1}(\bm{\mu}_\alpha))$,
which, since $\bm{\mu}_* -  \bm{\mu} = \epsilon \bm{\delta}_{\bm{\mu}}$, yields
the influence function of $\bm{x}_*$:
\begin{eqnarray}
\bm{\delta}_{\bm{\mu}} & = & \frac{g(\bm{x}_*)}{\Sigma_*} \times J^{-1}_u(u^{-1}(\bm{\mu}_\alpha)) \left(u(\bm{x}_*) - u(\bm{\mu})\right)\:\:.\label{dedelta}
\end{eqnarray}
Let $J_u$ denote
$J_u(u^{-1}(\bm{\mu}_\alpha))$ for short. Eq. (\ref{dedelta}) brings
\begin{eqnarray}
\|\bm{\delta}_{\bm{\mu}}\|_2^2 & = & \bm{\delta}^\top (J_u^{-1})^\top
J_u^{-1}\bm{\delta}\leq \|\bm{\delta}\|_2^2
\lambda_{\mathrm{max}}\:\:,\label{bdelt2}
\end{eqnarray}
for some upperbound $\lambda_{\mathrm{max}}$ on the eigenvalues of $(J_u^{-1})^\top
J_u^{-1}$. We also have
\begin{eqnarray*}
\|\bm{\delta}\|_2^2 & = & \frac{g(\bm{x}_*)}{\Sigma_*} \times \left\|
  u(\bm{x}_*) - u(\bm{\mu})\right\|_2^2\\
 & \leq & \frac{2 g(\bm{x}_*)}{\Sigma_*} \times (\|u(\bm{x}_*)\|_2^2 +
 \|u(\bm{\mu})\|_2^2)\\
 & & = \underbrace{\frac{2 g(\bm{x}_*) \|u(\bm{x}_*)\|_2^2}{\Sigma_*}}_{a} + \underbrace{\frac{2 g(\bm{x}_*) \|u(\bm{\mu})\|_2^2}{\Sigma_*}}_{b}\:\:.
\end{eqnarray*}
Because $\epsilon < 1-\tau$, we have $\Sigma_* \geq \tau
g(\bm{x}_1) \defeq K_1$ for
some $K_1>0$ which does not depend upon $\bm{x}_*$ or $\epsilon$. Because of assumption
(i), $g(\bm{x}_*) \leq K_2$ for some constant $K_2>0$ and so $b =
\tilde{O}(1)$ (the tilda meaning that the function does not depend
upon $\bm{x}_*$ or $\epsilon$). Because of assumption (ii),
$g(\bm{x}_*)\|u(\bm{x}_*)\|_2 \leq K_3$ for some constant $K_3>0$ and
so $a=\tilde{O}(1)$. Hence, $\|\bm{\delta}\|_2^2 = \tilde{O}(1)$. There remains to
plug this into (\ref{bdelt2}), and remark that $\lambda_{\mathrm{max}}
\leq 1/\lambda$ from assumption (iii), to conclude. 

\section{Proof of Lemma \ref{robr}}\label{subrobr}

The proof of this Lemma relies on the following Taylor expansions:
\begin{eqnarray}
v(\bm{\mu} + \epsilon \bm{\delta}_{\bm{\mu}}) & = & v(\bm{\mu}) + \epsilon
J_v \bm{\delta}_{\bm{\mu}}\:\:,\label{tayle1}
\end{eqnarray}
for some value $J_v$ of the Jacobian of $v$ in between $\bm{\mu}$ and
$\bm{\mu} + \epsilon \bm{\delta}_{\bm{\mu}}$, 
\begin{eqnarray}
v(\bm{x}_*) & = & v(\bm{\mu}) + J'_v (\bm{x}_*-\bm{\mu})\:\:,\label{tayle1b}
\end{eqnarray}
for some value $J'_v$ of the Jacobian of $v$ in between $\bm{\mu}$ and
$\bm{x}_*$, and
{\footnotesize
\begin{eqnarray}
\lefteqn{\varphi(v(\bm{\mu} + \epsilon \bm{\delta}_{\bm{\mu}}))}\nonumber\\
 & = &
\varphi(v(\bm{\mu})) + (v(\bm{\mu} + \epsilon \bm{\delta}_{\bm{\mu}}) -
v(\bm{\mu}))^\top \nabla\varphi(v(\bm{\mu}))\nonumber\\
 & & + \frac{1}{2} (v(\bm{\mu} + \epsilon \bm{\delta}_{\bm{\mu}}) -
v(\bm{\mu}))^\top \mathrm{H}_1 (v(\bm{\mu} + \epsilon \bm{\delta}_{\bm{\mu}}) -
v(\bm{\mu}))\nonumber\\
 & = & \varphi(v(\bm{\mu})) + \epsilon \bm{\delta}_{\bm{\mu}}^\top
 J^\top_v \nabla\varphi(v(\bm{\mu}))\nonumber\\
 & & +\frac{\epsilon^2}{2} \bm{\delta}_{\bm{\mu}}^\top
 J^\top_v \mathrm{H}_1 J_v
 \bm{\delta}_{\bm{\mu}}\label{tayle3a}\\
 & = & \varphi(v(\bm{\mu})) + \epsilon \bm{\delta}_{\bm{\mu}}^\top
 J^\top_v u(\bm{\mu}) +\frac{\epsilon^2}{2} \bm{\delta}_{\bm{\mu}}^\top
 J^\top_v\mathrm{H}_1  J_v
 \bm{\delta}_{\bm{\mu}}\:\:,\label{tayle3b}
\end{eqnarray}
}
for some value $\mathrm{H}_1$ of the Hessian of $\mathrm{H} \varphi$
in between $v(\bm{\mu})$ and $v(\bm{\mu} + \epsilon
\bm{\delta}_{\bm{\mu}})$. We have made use of (\ref{tayle1}) in
(\ref{tayle3a}) and the fact that $\nabla\varphi \circ v = u$ in (\ref{tayle3b}).

According to Theorem \ref{zigenphi}, the population minimizers of
${\mathcal{S}}$ to which we add $\bm{\mu}$ with a weight of $\epsilon$
satisfy $\bm{\delta}^+_{v, \epsilon} \bot \bm{z}^+_u$, with
$\bm{\delta}^+_{v, \epsilon} \defeq \bm{x}^+_{v, \epsilon} -
\bm{\mu}^+_{v, *}$ and:
\begin{eqnarray}
\bm{x}^+_{v, \epsilon} \hspace{-0.3cm} & \defeq & (1-\epsilon) \bm{x}^+_v + \epsilon \left[
\begin{array}{c}
v(\bm{x}_*)\\
\varphi(v(\bm{x}_*))
\end{array}
\right] \nonumber\\
 & & = (1-\epsilon) \bm{x}_v^+ + \epsilon
\bm{x}_{v,*}^+\label{defxpe}\:\:,\\
\bm{\mu}^+_{v,*} \hspace{-0.3cm} & \defeq & \left[
\begin{array}{c}
v(\bm{\mu} + \epsilon \bm{\delta}_{\bm{\mu}})\\
\varphi(v(\bm{\mu} + \epsilon \bm{\delta}_{\bm{\mu}}))
\end{array}
\right] = \bm{\mu}_v^+ + \epsilon
 \bm{\delta}^+_{\bm{\mu}}\label{defmups}\:\:,\\
\bm{\delta}^+_{\bm{\mu}} & \defeq & \left[
\begin{array}{c}
J_v \bm{\delta}_{\bm{\mu}}\\
\bm{\delta}_{\bm{\mu}}^\top
 J^\top_v u(\bm{\mu})+\frac{\epsilon}{2} \bm{\delta}_{\bm{\mu}}^\top
 J^\top_v\mathrm{H}_1  J_v
 \bm{\delta}_{\bm{\mu}}
\end{array}
\right] \label{deffmu}\:\:.
\end{eqnarray}
(See (\ref{defdeltav}) for the definitions of
$\bm{x}_v^+, \bm{\mu}_v^+$) In (\ref{deffmu}), we have used
(\ref{tayle1}) and (\ref{tayle3b}). We obtain:
\begin{eqnarray}
\lefteqn{0 = (\bm{\delta}^+_{v, \epsilon})^\top \bm{z}^+_u}\nonumber\\
 & = & (1-\epsilon)
(\bm{x}_v^+ - \bm{\mu}^+_{v,*})^\top \bm{z}_u^+ + \epsilon (\bm{x}^+_{v,*} - \bm{\mu}^+_{v,*})^\top \bm{z}^+_u \nonumber\\
 & = & (1-\epsilon) (\bm{\delta}_v^+)^\top \bm{z}^+_u -
 \epsilon(1-\epsilon) (\bm{\delta}^+_{\bm{\mu}})^\top \bm{z}^+_u \nonumber \\
 & & + \epsilon (\bm{x}^+_{v,*} - \bm{\mu}^+_v)^\top \bm{z}^+_u - \epsilon^2
 (\bm{\delta}^+_{\bm{\mu}})^\top \bm{z}^+_u\nonumber\\ 
 & = & \epsilon ( (\bm{x}^+_{v,*} - \bm{\mu}^+_v)^\top \bm{z}^+_u -
 (\bm{\delta}^+_{\bm{\mu}})^\top \bm{z}^+_u)\:\:. \label{simp1}
\end{eqnarray}
In (\ref{simp1}), we have used the fact that $(\bm{\delta}^+_v)^\top \bm{z}^+_u = 0$ since $\bm{\mu}$ is a right population minimizer for
the $v$-conformal divergence on ${\mathcal{S}}$. Since $\epsilon \neq
0$, we obtain from (\ref{simp1}) the equation which is central to the proof of Lemma \ref{robr}:
\begin{eqnarray}
(\bm{\delta}^+_{\bm{\mu}})^\top \bm{z}^+_u & = & (\bm{x}^+_{v,*} -
\bm{\mu}^+_v)^\top \bm{z}^+_u\:\:.\label{pdef}
\end{eqnarray}
We now work on this equation. 
Looking at $\bm{z}^+_u$ in (\ref{defzv}), we observe that
$\|\bm{z}^+_u\|_2 \neq 0$, $\forall \bm{z} \neq \bm{0}$. To see this, for
$\|\bm{z}^+_u\|_2 = 0$, we first need $\nabla f(u(\bm{\mu}))^\top
\bm{z} = 0$, and this implies $f(u(\bm{\mu})) \times \bm{z} = 0$, which implies
$\bm{z} = \bm{0}$. So, assuming that we pick $\bm{z} \neq \bm{0}$, we can simplify (\ref{pdef}) and
obtain 
\begin{eqnarray}
\|\bm{\delta}^+_{\bm{\mu}}\|_2^2 & \leq & \frac{1}{\cos^2(\bm{\delta}^+_{\bm{\mu}}, \bm{z}^+_u)} \times \|\bm{x}^+_{v,*} -
\bm{\mu}^+_v\|_2^2\:\:, \forall
\bm{z} \neq \bm{0}\:\:. \label{fzeq}
\end{eqnarray}
We now find a $\bm{z} \neq \bm{0}$ with which the inverse square
cosine is small. To find this $\bm{z}$, we use
this intermediate result, (\textbf{P}):
\begin{enumerate}
\item [(\textbf{P})] let $\bm{a}, \bm{b}, \bm{c} \in {\mathbb{R}}^d$.  The solution $\bm{z}$ to the equation
$\bm{z} = \bm{a} + \bm{c}^\top \bm{z} \times \bm{b}$ is:
{\footnotesize
\begin{eqnarray*}
\bm{z} & = & \left\{
\begin{array}{ccl}
\bm{a} + \frac{\bm{c}^\top \bm{a}}{1-\bm{c}^\top \bm{b}}\times \bm{b}
& \mbox{ if } & \bm{c} \neq \bm{0} \wedge \bm{b}
\neq \|\bm{c}\|_2^{-2}\bm{c}\\
\bm{a} & \mbox{ if } & \bm{c} = \bm{0} 
\end{array}\right.\:\:.
\end{eqnarray*}
}
\end{enumerate}
We use (\textbf{P}) with the following vectors:
\begin{eqnarray*}
\bm{a} & \defeq & \frac{1}{f(u(\bm{\mu}))}\times J_v
\bm{\delta}_{\bm{\mu}}\:\:,\\
\bm{b} & \defeq & -\frac{1}{f(u(\bm{\mu}))}\times u(\bm{\mu})\:\:;\\
\bm{c} & \defeq & \nabla f(u(\bm{\mu}))\:\:.
\end{eqnarray*}
When $\|\nabla f(u(\bm{\mu}))\|_2 \neq 0$, we need to check if it can
be possible that $\bm{b}
= \|\bm{c}\|_2^{-2}\bm{c}$. For this to happen from the definitions of
$\bm{b}$ and $\bm{c}$, we need $\nabla
f= \alpha \mathrm{Id}$ for some $\alpha\neq 0$, implying $f(\bm{z}) =
\alpha\|\bm{z}\|^2_2/2$, which cannot be the case from assumptions (i) and (ii).

Let us analyze the two cases of (\textbf{P}), starting from the case
$\bm{c} \neq \bm{0}$. We have:
{\footnotesize
\begin{eqnarray*}
\bm{z} & = & \frac{1}{f(u(\bm{\mu}))}\nonumber\\
 & & \times\left(  J_v
\bm{\delta}_{\bm{\mu}} - \frac{\nabla f(u(\bm{\mu}))^\top J_v
\bm{\delta}_{\bm{\mu}}}{f(u(\bm{\mu}))+\nabla f(u(\bm{\mu}))^\top u(\bm{\mu})}\times u(\bm{\mu})\right)\:\:,
\end{eqnarray*}
}
which yields:
\begin{eqnarray}
\bm{z}^+_u & \defeq & \left[
\begin{array}{c}
J_v \bm{\delta}_{\bm{\mu}}\\
-\frac{\nabla
  f(u(\bm{\mu}))^\top  J_v
\bm{\delta}_{\bm{\mu}}}{f(u(\bm{\mu}))+\nabla f(u(\bm{\mu}))^\top u(\bm{\mu})}
\end{array}
\right] \:\:.\label{defzv2}
\end{eqnarray}
Let us define:
\begin{eqnarray*}
x & \defeq & \frac{1}{\|J_v \bm{\delta}_{\bm{\mu}}\|_2}\times \frac{\nabla
  f(u(\bm{\mu}))^\top  J_v
\bm{\delta}_{\bm{\mu}}}{f(u(\bm{\mu}))+\nabla f(u(\bm{\mu}))^\top
u(\bm{\mu})}\:\:,\\
y & \defeq & \frac{1}{\|J_v \bm{\delta}_{\bm{\mu}}\|_2}\times \left(\bm{\delta}_{\bm{\mu}}^\top
 J^\top_v u(\bm{\mu})+\frac{\epsilon}{2} \bm{\delta}_{\bm{\mu}}^\top
 J^\top_v\mathrm{H}_1  J_v
 \bm{\delta}_{\bm{\mu}} \right)\:\:.
\end{eqnarray*}
Plugging the expression of $\bm{z}^+_u$ in
$1/\cos^2(\bm{\delta}^+_{\bm{\mu}}, \bm{z}^+_u)$, we obtain after simplification:
\begin{eqnarray}
\frac{1}{\cos^2(\bm{\delta}^+_{\bm{\mu}}, \bm{z}^+_u)} & = & 1+
\frac{(x+y)^2}{1+x^2+y^2+x^2y^2} \nonumber\\
 & \leq & 2 \:\:.\label{pcos}
\end{eqnarray}
We obtain the following upperbound on
$\|\bm{\delta}^+_{\bm{\mu}}\|_2^2$ refined from (\ref{fzeq}):
\begin{eqnarray}
\|\bm{\delta}^+_{\bm{\mu}}\|_2^2 & \leq & 2 \|\bm{x}^+_{v,*} -
\bm{\mu}^+_v\|_2^2\:\:. \label{fzeq2}
\end{eqnarray}
Handling the second case for (\textbf{P}) is simpler, as since
$\bm{c} = \nabla f(u(\bm{\mu})) = \bm{0}$, picking $\bm{z} =
(1/f(u(\bm{\mu}))) \times J_v
\bm{\delta}_{\bm{\mu}}$ yields (\ref{pcos}) with $x=0$, and
(\ref{fzeq2}) is still valid. To finish up with the proof, we first
upperbound the right-hand side of (\ref{fzeq2}), as:
\begin{eqnarray}
\lefteqn{\|\bm{x}^+_{v,*} -
\bm{\mu}^+_v\|_2^2}\nonumber\\
 & = & \|v(\bm{x}_*) -
v(\bm{\mu})\|_2^2 +
(\varphi(v(\bm{x}_*))-\varphi(v(\bm{\mu})))^2\nonumber\\
 & \leq & \|v(\bm{x}_*) -
v(\bm{\mu})\|_2^2(1+L^2)\label{fineq1}\\
 &  & =   (\bm{x}_*-\bm{\mu})^\top (J'_v)^\top J'_v
 (\bm{x}_*-\bm{\mu}) (1+L^2) \label{fineq2}\\
 & \leq & \lambda_{\mathrm{\max}} (1+L^2) \|\bm{x}_*-\bm{\mu}\|_2^2 \label{fineq3}\:\:.
\end{eqnarray}
Ineq. (\ref{fineq1}) follows from Definition \ref{dwrob}, (\ref{fineq2})
follows from (\ref{tayle1b}), and (\ref{fineq3}) is obtained using a
finite upperbound $\lambda_{\mathrm{\max}}$ for the
eigenvalues 
of $J_v^\top
J_v$. We then lowerbound the left-hand side of
(\ref{fzeq2}) as:
\begin{eqnarray}
\|\bm{\delta}^+_{\bm{\mu}}\|_2^2 & = & (\bm{\delta}_{\bm{\mu}}^\top
 J^\top_v J_v
 \bm{\delta}_{\bm{\mu}}) \nonumber\\
 & &+ (\bm{\delta}_{\bm{\mu}}^\top
 J^\top_v u(\bm{\mu})+\frac{\epsilon}{2} \bm{\delta}_{\bm{\mu}}^\top
 J^\top_v\mathrm{H}_1  J_v
 \bm{\delta}_{\bm{\mu}})^2\nonumber\\
 & \geq & \bm{\delta}_{\bm{\mu}}^\top
 J^\top_v J_v
 \bm{\delta}_{\bm{\mu}}\nonumber\\
 & \geq & \lambda_{\mathrm{\min}} \|\bm{\delta}_{\bm{\mu}}\|_2^2 \:\:,\label{fineq4}
\end{eqnarray}
for some non-zero lowerbound $\lambda_{\mathrm{\min}}$ of the eigenvalues of $J_v^\top
J_v$. We then obtain from (\ref{fzeq2}),
(\ref{fineq3}) and (\ref{fineq4}):
\begin{eqnarray*}
\|\bm{\delta}_{\bm{\mu}}\|_2^2 & \leq & \frac{2
  \lambda_{\mathrm{\max}}(1+L^2)}{\lambda_{\mathrm{\min}}}
\times \|\bm{x}_*-\bm{\mu}\|_2^2\:\:.\\
 & \leq & 2\lambda_v (1+L^2) \|\bm{x}_*-\bm{\mu}\|_2^2\:\:,
\end{eqnarray*}
using the definition of $\lambda_v$. We obtain
$\|\bm{\delta}_{\bm{\mu}}\|_2 \leq \sqrt{\lambda_v} \ell(L)
\|\bm{x}_*-\bm{\mu}\|_2$ for $\ell(L) \defeq \sqrt{2}(1+L) \geq
\sqrt{2(1+L^2)}$, as claimed.

\section{Proof of Corollary \ref{csrpm}}\label{sec-csrpm}

(Of point 1)) ($\Rightarrow$) $\bm{\mu} = \overline{\bm{x}}$ zeroes the right-hand side of
(\ref{eqfon2}), which, since $\overline{\varphi} \neq
\varphi(\bm{\mu})$, implies $\mathrm{D}_{\bm{z}} g(\bm{\mu}) = 0$, for
any $\bm{z} \neq \bm{0}$, and so $g(\bm{\mu})$ is
constant. ($\Leftarrow$) is a property of ordinary Bregman divergences.

(Of point 2))  ($\Rightarrow$) This time, $\varphi(\bm{\mu}) = \overline{\varphi}$ zeroes the
left-hand side of
(\ref{eqfon2}). Because $\varphi$ is strictly convex,
$\overline{\bm{x}} \neq \bm{\mu}$ and so (\ref{eqfon2}) brings $\mathrm{D}_{\bm{z}} (g(\bm{\mu})\nabla \varphi(\bm{\mu})) =
\bm{0}, \forall \bm{z}$, and so $\nabla \varphi(\bm{\mu}) =
(K/g(\bm{\mu})) \bm{u}$ for some constants $K$ and vector $\bm{u}$. The hessian coordinates
are $\mathrm{H}_{ij}\varphi(\bm{\mu}) = -(K u^i/g^2(\bm{\mu}))\partial
g(\bm{\mu})/\partial \mu_j$. Because the Hessian is symmetric, we obtain $u^j\partial
g(\bm{\mu})/\partial \mu_i = u^i\partial
g(\bm{\mu})/\partial \mu_j$, and so $g(\bm{\mu})$ can be expressed as
$g(\bm{\mu}) = h( \bm{u}^\top \bm{\mu})$ for some function $h :
{\mathbb{R}} \rightarrow {\mathbb{R}}$. We
get $\bm{x}^\top \mathrm{H}\varphi \bm{x} = -(Kh'(\bm{u}^\top
\bm{\mu})/h^2(\bm{u}^\top \bm{\mu})) \|\bm{u} \bullet \bm{x}\|_2^2$,
with ''$\bullet$" denoting Hadamard product, and since we want
$\bm{x}^\top \mathrm{H}\varphi \bm{x} > 0$ when $\bm{x} \neq \bm{0}$, $h'$ has to be of a different sign than $K$.
 ($\Leftarrow$) is immediate.

\section{Proof of Lemma \ref{lsd}}\label{sec-lsd}

First, since
$\varphi' = v \circ u^{-1}$, we get 
\begin{eqnarray}
\varphi'' & = & (v' \circ
u^{-1})/(u' \circ u^{-1})\:\:,\label{defphi2}
\end{eqnarray}
and so \textbf{(ii)} would be a consequence of \textbf{(iii)}. Let us compute the equality
of partial
derivatives in $x$ of (\ref{eq00}):
\begin{eqnarray*}
\lefteqn{g(u(y)) v'(x) (u(x)-u(y))}\\
& = & u'(x)g'(u(x))D_\varphi(v(y):v(x)) \\
 & & -
  g(u(x))u'(x)(v(y)-v(x)) \:\:.
\end{eqnarray*}
We then compute the partial derivatives in $y$ and reorganize:
{\footnotesize
\begin{eqnarray}
\lefteqn{u'(x) v'(y) [g'(u(x))(u(y)-u(x)) - g(u(x))]} \nonumber\\
& = & u'(y) v'(x) [g'(u(y))(u(x)-u(y)) - g(u(y))] \label{pqr1}\:\:.
\end{eqnarray}
}
We now use (\ref{defphi2}), letting $z\defeq u(x)$ and $r \defeq
u(y)$, so that (\ref{pqr1}) becomes:
\begin{eqnarray}
\lefteqn{\varphi''(r)[g'(z)(r-z) - g(z)]}\nonumber\\
 & = & \varphi''(z)[g'(r)(z-r) - g(r)]\:\:.\label{pgg}
\end{eqnarray}
Let us fix temporarily $z$ to a constant, so that (\ref{pgg}) is a function of
$r$, and thus reads:
{\footnotesize
\begin{eqnarray}
g'(z)(r-z) - g(z) & = & \varphi''(z)\times \vartheta(r,z)\:\:, \label{def0}\\
\vartheta(r,z) & \defeq & \frac{g'(r)(z-r) - g(r)
}{\varphi''(r)}\nonumber\\
 & = & \frac{g'(r)}{\varphi''(r)} z - \frac{g(r) + rg'(r)}{\varphi''(r)}\:\:. \label{def1}
\end{eqnarray}
}
Because the left hand-side of
(\ref{def0}) is linear in $r$, so has to be $\vartheta$ in (\ref{def1}),
and so we get:
\begin{eqnarray}
\frac{g'(r)}{\varphi''(r)} & = & ar + b\:\:, \label{gg1}\\
\frac{g(r) + rg'(r)}{\varphi''(r)} & = & cr + d\:\:, \label{gg2}
\end{eqnarray}
for some $a,b,c,d \in {\mathbb{R}}$ that are constant since $z$ is fixed; our objective is to
prove that all but $d$ are zero, so let us proceed by
assuming that all are non zero. Substituting $g'(r)$ from
(\ref{gg1}) in (\ref{gg2}) yields:
\begin{eqnarray}
\varphi''(r) & = & \frac{g(r)}{f_1(r)} \:\:,\\
f_1(r) & \defeq & -ar^2 + (c-b)r+ d\:\:.
\end{eqnarray}
Since $a\neq 0$, $f_1$ is the equation of a parabola.
Using (\ref{gg1}), we see that $g$ is solution of the following
homogeneous differential equation: 
\begin{eqnarray}
(ar+b)g(r) - f_1(r) g'(r) & = & 0\:\:,\label{eqdi}
\end{eqnarray}
whose solution is found to be, for any constant $K_5$:
\begin{eqnarray}
\lefteqn{g(r)}\nonumber\\
 & = & K_5 \exp\left(\int{\frac{ar+b}{f_1(r)}}\right)\nonumber\\
 & = & \frac{K_5}{\sqrt{-f_1(r)}} \left(\sqrt{\frac{\sqrt{K_6} +
     f_2(r)}{\sqrt{K_6} -
     f_2(r)}}\right)^{\frac{f_2(c/a)}{\sqrt{K_6}}} \:\:;\label{fgg}\\
 K_6 & \defeq &  4 ad + (b-c)^2\:\:;\nonumber\\
f_2(r) & \defeq & 2ar+(b-c)\:\:.\nonumber
\end{eqnarray}
Eq. (\ref{fgg}) implies $K_6 \geq 0$. For $g$ as in (\ref{fgg}) to exist, we have two more constraints to
meet: \textbf{(a)} $-f_1(r) > 0$ and \textbf{(b)} $(\sqrt{K_6} +
     f_2(r)) / (\sqrt{K_6} -
     f_2(r)) \geq 0$. We distinguish two cases:
\begin{itemize}
\item[] ($a>0$) To meet constraint \textbf{(a)}, we need $r >
((c-b) + \sqrt{K_6})/(2a)$ or $r <
((c-b) - \sqrt{K_6})/(2a)$. In both cases, constraint \textbf{(b)} is
violated as respectively the denominator or the numerator (only) of the
fraction is strictly negative.
\item[] ($a<0$) To meet constraint \textbf{(a)}, we need $((c-b)
- \sqrt{K_6})/(2a) < r < ((c-b)
+ \sqrt{K_6})/(2a)$. Again, constraint \textbf{(b)} is violated.
\end{itemize}
We end up with the conclusion that $a=0$ so that $f_1$ is
linear. Assume now that $c\neq b$. The new solution to (\ref{eqdi}) is:
\begin{eqnarray*}
g(r) & = & K_5 \left((c-b)r+d\right)^{\frac{b}{c-b}}\:\:,
\end{eqnarray*}
leading through (\ref{gg1}) to:
\begin{eqnarray}
\varphi''(r) & = & K_5 \left((c-b)r+d\right)^{\frac{b}{c-b}-1}\:\:.
\end{eqnarray}
This enforces $K_5>0$, but $\varphi''(-d/(c-b)) = 0$,
which is not possible as $\varphi$ must be strictly convex. Hence
$a=0$ and $c=b$, so that $f_1(r) = d = f_1$ is a constant. 

To finish up the proof, we
consider the assumption
$b\neq 0$. The new solution to (\ref{eqdi}) is:
\begin{eqnarray}
g(r) & = & K_5 \exp(br/d)\:\:,\label{ggg1}
\end{eqnarray}
leading through (\ref{gg1}) to:
\begin{eqnarray}
\varphi''(r) & = & \frac{K_5}{d} \exp(br/d)\:\:,\label{ggg2}
\end{eqnarray}
enforcing this time $K_5/d > 0$. 

Now, let us start back from (\ref{pgg}), considering
$b,c,d$ functions of $z$. We simplify (\ref{pgg}) using (\ref{gg1}),
(\ref{gg2}) and the expressions of $g$ and $\varphi''$ in (\ref{ggg1}) and (\ref{ggg2}), and obtain $b(r)r-(c(r)z+d(r)) = b(z)z - (c(z)r + d(z))$,
that is, since $c=b$:
\begin{eqnarray}
b(r)(r-z) - d(r) & = & b(z)(z-r) - d(z)\:\:,\label{pg1}
\end{eqnarray}
or, equivalently, for $z\neq r$,
\begin{eqnarray}
\frac{d(z) - d(r)}{z-r} & = & b(z) + b(r)\:\:.
\end{eqnarray}
This shows that $d$ is derivable, and its derivative satisfies $d'(z)
= 2b(z)$, and so:
\begin{eqnarray}
d(z) & = & 2\int b(z) + K_7\:\:, \label{defd}
\end{eqnarray}
for any constant $K_7$. We put this expression in (\ref{pg1}), differentiate in $r$ and
obtain $b'(r)(r-z)+b(r)-2b(r) = -b(z)$, that is, after reordering,
$b(z) = b(r) + (z-r)b'(r)$, and so:
\begin{eqnarray}
b(z) & = & K_8z+K_9 \:\:,
\end{eqnarray} 
for any
constants $K_8$ and $K_9$. Plugging this in (\ref{pg1}) using
(\ref{defd}) yields the identity, valid for any $z$ and $r$: $(
K_8r+K_9)(r-z) - (K_8r^2/2+K_9r+K_2) = ( K_8z+K_9)(z-r) -
(K_8z^2/2+K_9z+K_2)$. Its simplification yields:
\begin{eqnarray}
K_8(z-r)(z+r) & = & 2 K_9 (z-r), \forall z, r\:\:.
\end{eqnarray}
This implies $K_8 = K_9 = 0$, and finally the solutions to (\ref{gg1})
and (\ref{gg2}) are $a=b=c=0$ and $d = K_7$, constant. We obtain $g(r)$
constant as in \textbf{(i)} through (\ref{gg1}) and $\varphi''$ constant through (\ref{gg2})
--- and the expression of $\varphi$ as in \textbf{(ii)} ---. Finally,
since $\varphi' = v\circ u^{-1}$, it comes $v\circ u^{-1} = \ell(x)$,
and so $v = \ell(u)$, as claimed.

\end{document}